\documentclass[final,3p,times]{elsarticle}

\usepackage{amssymb}
\usepackage{amsthm}
\usepackage[cmex10]{amsmath}

\usepackage{graphicx}
\usepackage{float}
\usepackage{subfig}
\usepackage[]{placeins}
\usepackage[]{algorithm}
\usepackage{algpseudocode}
\usepackage{multirow}
\usepackage[subnum]{cases}
\usepackage{tabularx}

\usepackage[]{hyperref} 
\usepackage[]{amsrefs} 

\newcommand\noi{\noindent}

\renewcommand{\algorithmicforall}{\textbf{for each}}

\makeatletter
\renewcommand\appendix{\par
  \setcounter{section}{0}%
  \setcounter{subsection}{0}%
  \setcounter{equation}{0}%
  \setcounter{table}{0}
  \setcounter{figure}{0}
  \gdef\theequation{\@Alph\c@section.\arabic{equation}}%
  \gdef\thefigure{\@Alph\c@section.\arabic{figure}}%
  \gdef\thetable{\@Alph\c@section.\arabic{table}}%
  \gdef\thesection{\Alph{section}}%
  \@addtoreset{equation}{section}%
  \@addtoreset{table}{section}
  \@addtoreset{figure}{section}
}
\makeatother

\begin{document}
\begin{frontmatter}
\title{A parallel space saving algorithm for frequent items  \\ and the Hurwitz zeta distribution}
\author [unile] {Massimo~Cafaro\corref{cor1}}
\ead{massimo.cafaro@unisalento.it}
\cortext[cor1]{Corresponding author}
\author [unile] {Marco Pulimeno}
\ead{marco.pulimeno@unisalento.it}
\author [comp] {Piergiulio Tempesta}
\ead{p.tempesta@fis.ucm.es}
\address[unile]{University of Salento, Lecce, Italy}
\address[comp]{Departamento de F\'{\i}sica Te\'{o}rica II, Facultad de F\'{\i}sicas, Universidad
Complutense, 28040 -- Madrid, Spain and Instituto de Ciencias Matem\'aticas, C/ Nicol\'as Cabrera, No 13--15, 28049 Madrid, Spain}

\begin{abstract} We present a message-passing based parallel version of the Space Saving algorithm designed to solve the $k$--majority problem. The algorithm determines in parallel frequent items, i.e., those whose frequency is greater than a given threshold, and is therefore useful for iceberg queries and many other different contexts. We apply our algorithm to the detection of frequent items in both real and synthetic datasets whose probability distribution functions are a Hurwitz and a Zipf distribution respectively. Also, we compare its parallel performances and accuracy against a parallel algorithm recently proposed for merging summaries derived by the Space Saving or Frequent algorithms.
\end{abstract}


\newtheorem{thm}{Theorem}
\newtheorem{lem}[thm]{Lemma}
\newdefinition{rmk}{Remark}
\newproof{pf}{Proof}
\newtheorem{prop}[thm]{Proposition}
\newtheorem*{cor}{Corollary}

\newdefinition{defn}{Definition}
\newtheorem{conj}{Conjecture}
\newtheorem{exmp}{Example}
\newtheorem{case}{Case}

\end{frontmatter}


\section{Introduction}
\label{intro}
Discovering frequent items is a  data mining problem that attracted many researchers, owing to its relevance to  applications in  several domains. The problem is also known in the literature, depending on the specific application, as \textit{hot list analysis} \cite{Gibbons}, market basket analysis \cite{Brin} and  \textit{iceberg query} \cite{Fang98computingiceberg}, \cite{Beyer99bottom-upcomputation}. Additional applications include network traffic analysis \cite{DemaineLM02},  \cite{Estan}, \cite{Pan}, the analysis of web logs \cite{Charikar}, Computational and theoretical Linguistics \cite{CICLing}, 
 ecological field studies \cite{Mouillot}, etc.
Several sequential solutions have been provided. In their survey \cite{Cormode}, Cormode and Hadjieleftheriou classify existing algorithms as being either \emph{counter} or \emph{sketch} based. Misra and Gries \cite{Misra82} proposed the first counters--based sequential algorithm, which has been rediscovered independently by Demaine et al. \cite{DemaineLM02} and Karp et al. \cite{Karp}. Recently designed counters--based algorithms include \emph{LossyCounting}
\cite{Manku02approximatefrequency} and \emph{Space Saving} \cite{Metwally2006}. In particular, Space Saving has been shown to be the most efficient and accurate algorithm among counters--based ones \cite{Cormode}, which motivates our choice of designing a parallel version of this algorithm.
Notable sketch--based solutions are \emph{CountSketch} \cite{Charikar} and \emph{CountMin} \cite{Cormode05}. In the parallel setting, we presented in \cite{cafaro-tempesta} a message-passing based parallel version of the \emph{Frequent} algorithm, whilst \cite{Zhang2013} presents a shared-memory parallel version.  A parallel version of the Lossy Counting algorithm has been proposed in \cite{Zhang2012}. Parallel versions of the Space Saving algorithm for shared-memory architectures have been designed in \cite{Roy2012} and \cite{Das2009}. A GPU (Graphics Processing Unit) accelerated algorithm for frequent items appeared in \cite{Govindaraju2005} and \cite{Erra2012}. 
Novel shared-memory parallel algorithms for frequent items were recently proposed in \cite{Tangwongsan2014}.

A similar problem, mining frequent itemsets, is strongly related to the problem of association pattern mining, which originated from the analysis of market basket data. Such datasets are basically sets of items bought by customers, traditionally referred to as transactions. Association pattern mining entails discovering the so-called association rules between sets of items. Frequent itemsets (also known in the literature as frequent patterns) are those sets of items determined by the mining process; association rules are simple implications stated as $A \implies B$, in which both $A$ and $B$ are sets of items. Among the many algorithms appeared in the literature, we recall here recent work including \cite{fis1}, \cite{fis2}, \cite{fis3} and \cite{fis4}.

In this paper, we investigate how to parallelize the \emph{Space Saving} algorithm, and we design our algorithm in the context of message--passing architectures. To the best of our knowledge, this is the first parallel version of the Space Saving algorithm for message-passing architectures. Therefore it is the only one that can solve arbitrarily large problems on millions of cores, owing to intrinsic hardware limits related to shared-memory architectures preventing scalability of SMP (Symmetric Multi-Processing) nodes to higher processor counts. Indeed, current SMP nodes are equipped with at most a few dozens of processors. We prove the correctness of the algorithm, and then analyze its parallel complexity proving its cost--optimality for $k = O(1)$.

Another original aspect of this work is that we apply our algorithm to the study of frequent items in datasets whose probability distribution function is a \emph{Hurwitz distribution}. This distribution generalizes the classical Zipf distribution and is based on a well-known generalization of the Riemann zeta function, i.e. the Hurwitz function. We shall show that our parallel algorithm is especially suitable for treating these kind of datasets. We stress that the relevance of the Hurwitz distribution is very general. Indeed, the presence of an extra parameter makes it a more flexible tool than the classical Zipf one. To the best of our knowledge, this work offers the first example of application of the Hurwitz distribution in dataset analysis.

Before stating the problem solved by our algorithm, we need to recall a few basic definitions from multiset theory \cite{Syropoulos01}. We shall use a calligraphic capital letter to denote a multiset, and the corresponding capital Greek letter to denote its \textit{underlying} set. We are given a dataset $\mathcal{N}$ consisting of $n$ elements, and an integer $k$, with $2 \leq k \leq n$.

\begin{defn}
\label{multiset}
A multiset $\mathcal{N}=(N, f_{\mathcal{N}})$ is a pair where $N$ is some set, called the underlying set of elements, and $f_{\mathcal{N}}: N \rightarrow \mathbb{N}$ is a function.
The generalized indicator function of $\mathcal{N}$ is 
\begin{equation}
\label{eq01}
I_\mathcal{N} (x) := \left\{ {\begin{array}{*{20}c}
   {f_{\mathcal{N}}(x)} & {x \in N} , \\
   0 & {x \notin N},  \\
 \end{array} }\right.
 \end{equation}
 
\noindent where the integer--valued function $f_{\mathcal{N}}$, for each $x \in N$, provides its  \textit{frequency} (or multiplicity), i.e., the number of occurrences of $x$ in $\mathcal{N}$. 
\noi The cardinality of $\mathcal{N}$ is expressed by

\begin{equation}
\label{eq02}
\left\vert{\mathcal{N}}\right\vert := Card(\mathcal{N}) = \sum\limits_{x \in N} {I_\mathcal{N} (x)},
\end{equation}

\noi whilst the cardinality of the underlying set $N$ is

\begin{equation}
\label{eq03}
\left\vert{N}\right\vert := Card(N) = \sum\limits_{x \in N} {1}.
\end{equation}
\end{defn}

\noindent A multiset (also called a \emph{bag}) essentially is a set where the duplication of elements is allowed. In the sequel, $\mathcal{N}$ will play the role of a finite
input array, containing $n$ elements.

We can now state the problem formally.

\begin{defn}
\label{def2}
Given a \emph{multiset} $\mathcal{N}$, with $\left\vert{\mathcal{N}}\right\vert=n$, a $k$--majority element (or \emph{frequent item}) is an element $x\in N$ whose \emph{frequency} $f_{\mathcal{N}}(x)$ is such
that  $f_{\mathcal{N}}(x) \geq \left\lfloor {\frac{n}{k}} \right\rfloor+1$.
\end{defn}

\textbf{Statement of the problem}. \textit{The $k$--majority problem takes as input an array $\mathcal{N}$ of $n$ numbers, and requires as output the set $W = \left\{x \in N: f_{\mathcal{N}}(x)
\geq \left\lfloor {\frac{n}{k}} \right\rfloor+1\right\}$}.

Therefore, the $k$--majority problem entails finding the set of elements whose frequency is greater than a given threshold controlled by the parameter $k$. It is worth noting here that when $k = 2$, the problem reduces to the well known majority problem \cite{Moore81}, \cite{Fischer82}.

This article is organized as follows. We recall the sequential Space Saving algorithm in Section \ref{spacesaving}. Our parallel space saving algorithm is presented in Section \ref{pss-algorithm}. We prove its correctness in Section \ref{correctness},
analyze it and prove its cost--optimality for $k = O(1)$ in Section \ref{pss-analysis}. We provide and discuss in Appendix experimental results concerning the application of our algorithm to both real and synthetic datasets governed by a Zipf--Mandelbrot and by a Hurwitz distribution. In particular, we also compare our algorithm with another parallel algorithm designed and implemented by us starting from a sequential algorithm by Agarwal et al \cite{Agarwal}. Finally, we draw our conclusions in Section \ref{conclusions}.

\section{The space saving algorithm}
\label{spacesaving}

We recall here a few basic facts related to the sequential Space Saving algorithm that will be used later. The algorithm uses exactly $k$ counters in order to solve the $k$-majority problem sequentially, and allows estimating the maximum error committed when computing the frequency of an item. 
Space Saving works as described by the pseudocode of Algorithm \ref{ss}. We denote by $\mathcal{S}[i].e$, $\mathcal{S}[i].\hat{f}$ and $\mathcal{S}[i].\hat{\varepsilon}$ respectively the element monitored by the $i$th counter of $\mathcal{S}$, the corresponding estimated frequency and error committed in the estimation. When processing an item which is already monitored by a counter, its estimated frequency is incremented by one. When processing an item which is not already monitored by one of the available counters, there are two possibilities. If a counter is available, it will be in charge of monitoring the item and its estimated frequency is set to one.  Otherwise, if all of the counters are already occupied (their frequencies are different from zero), the counter storing the item with minimum frequency is incremented by one. Then the monitored item is evicted from the counter and replaced by the new item. This happens since an item which is not monitored can not have occurred in the input a number of times greater than the minimal frequency. The algorithm assumes that the item has occurred exactly a number of times equal to the frequency stored by the minimum counter, estimating by excess its frequency and introducing an error which is at most the minimum frequency. We keep track of this error, as done in \cite{Metwally2006}, by storing for each monitored item its error $\hat{\varepsilon}$.

Let $\mathcal{N} = (N, f_{\mathcal{N}})$ be the input multiset, $\mathcal{S} = (\Sigma, \hat{f_{\mathcal{S}}})$ the multiset of all of the monitored items and their respective counters at the end of the sequential Space Saving algorithm's execution, i.e., the algorithm's summary data structure. Let  $\left\vert \mathcal{S} \right\vert$ be the sum of the frequencies stored in the counters, $f_\mathcal{N}(e)$ the exact frequency of an item $e$, $\hat{f_\mathcal{S}}(e)$ its estimated frequency, $\hat{f_\mathcal{S}}^{min}$ the minimum frequency in $\mathcal{S}$ and $\hat{\varepsilon}_\mathcal{S}(e)$ the error of item $e$, i.e. an over-estimation of the difference between the estimated and exact frequency. It is worth noting here that $\hat{f_\mathcal{S}}^{min} = 0$ when  $\left\vert{\Sigma}\right\vert < k$. 
The following relations hold (as proved in \cite{Metwally2006}) for each item $e \in N$:

\begin{equation}
\label{eq04}
\left\vert{\mathcal{S}}\right\vert = \left\vert{\mathcal{N}}\right\vert,
\end{equation}

\begin{equation}
\label{eq05}
\hat{f_\mathcal{S}}(e) - \hat{f_\mathcal{S}}^{min} \leq\hat{f_\mathcal{S}}(e) - \hat{\varepsilon}_\mathcal{S}(e) \leq f_\mathcal{N}(e) \leq \hat{f_\mathcal{S}}(e),  \qquad e \in \Sigma,
\end{equation}

\begin{equation}
\label{eq06}
f_\mathcal{N}(e)  \leq \hat{f_\mathcal{S}}^{min}, \qquad \hspace{45mm} e \notin \Sigma,
\end{equation}

\begin{equation}
\label{eq07}
\hat{f_\mathcal{S}}^{min}  \leq \left\lfloor\frac{\left\vert{\mathcal{N}}\right\vert}{k}\right\rfloor.
\end{equation}

If an item $e$, at the end of the algorithm's execution, has an estimated frequency $\hat{f_\mathcal{S}}(e)$ less than the required threshold, $e$ can be excluded from the output, since it can not be a frequent item. Instead, if we keep track of the error $\hat{\varepsilon}_\mathcal{S}(e)$ and $\hat{f_\mathcal{S}}(e) - \hat{\varepsilon}_\mathcal{S}(e)$ is greater than or equal to the threshold, then $e$ is a frequent item. All of the other output items are only \emph{potential} frequent items.

\begin{algorithm}
\begin{algorithmic}[1]
\Require $\mathcal{N}$, an array; $start$, first index of $\mathcal{N}$ to be processed; $end$, last index of $\mathcal{N}$ to be processed; $k$, the $k$-majority parameter
\Ensure a summary containing $k$--majority candidate elements
\Procedure {SpaceSaving}{$\mathcal{N}, start, end, k$}
\State $\mathcal{S} \leftarrow \Call{InitializeCounters}{k}$
\For{$i=start$ to $end$}
	\If{$\mathcal{N}[i]$ is monitored}
		\State let $\mathcal{S}[l]$ be the counter of $\mathcal{N}[i]$ 
		\State $\mathcal{S}[l].\hat{f} \leftarrow \mathcal{S}[l].\hat{f} + 1$
	\Else
		\State let $\mathcal{S}[m].e$ be the element with least hits
		\State $\mathcal{S}[m].e \leftarrow \mathcal{N}[i]$ 
		\State $\mathcal{S}[m].\hat{\varepsilon} \leftarrow \mathcal{S}[m].\hat{f}$
		\State $\mathcal{S}[m].\hat{f} \leftarrow \mathcal{S}[m].\hat{f} + 1$
	\EndIf
\EndFor
\State \Return $\mathcal{S}$
\EndProcedure
\caption{Space saving.}
\label{ss}
\end{algorithmic}
\end{algorithm}

\section{A parallel space saving algorithm}

\label{pss-algorithm}

The pseudocode of Algorithm \ref{pss} describes our parallel Space Saving algorithm. We assume that the input array $\mathcal{N}$ is initially read by an application calling our function implementing the algorithm; for instance, every process reads the input from a file or a designated process reads it and broadcast it to the other processes. The initial call is \textit{ParallelSpaceSaving} $(\mathcal{N}, n, p, k)$, where $\mathcal{N}$ consists of $n$ elements, $p$ is the number of processors (or cores) we use in parallel and $k$ is the $k$-majority parameter. Each processor is assigned a unique rank; ranks are numbered from 0 to $p-1$. The algorithm determines in parallel $k$--majority candidates. We recall here that, indeed, some of the candidates returned may be false positives as in the sequential counterpart.

The algorithm works as follows. In the initial domain decomposition, each processor determines the indices of the first and last
element related to its block, by applying a simple block distribution, in which each processor is responsible for either $\left\lfloor {n/p} \right\rfloor$ or $\left\lceil
{n/p} \right\rceil$ elements.

Then, each processor determines $local$, a stream summary data structure storing its local candidates, their corresponding estimated frequencies and errors, by using the well-known  algorithm designed by Metwally et al. \cite{Metwally2006}, shown in the pseudocode as the \textit{SpaceSaving} function. An hash table $hash$ is then built, storing the local candidates as keys and their corresponding counters (estimated frequencies and errors) as values. This hash table is then sorted in ascending order by counters' frequency and used as input for the parallel reduction, whose purpose is to determine global candidates for the whole array.
This step is carried out by means of the $ParallelReduction$ function, shown as Algorithm. \ref{pssr}.

Assuming that the parallel reduction returns the result to the processor whose rank is 0, then that processor prunes the global candidates removing all of the items below the threshold required to be frequent items and returns the results. The $Pruned$ function, which is not shown here to save space, is just a linear scan in which every item's frequency is compared against the threshold and, if the frequency is greater than or equal to the threshold, then the item is returned in $result$ as a \emph{potential} frequent item.

The parallel reduction determines global candidates for the whole array and works as shown in Algorithm \ref{pssr}. In each sub-step of the reduction, a processor receives as input from two processors  $p_{r}$ and $p_{s}$ their  hash tables, that shall be called from now on  $\mathcal{S}_1$ and $\mathcal{S}_2$ respectively. These data structures contain local items as keys and their counters storing estimated frequencies and errors. For a generic summary $\mathcal{S}$, we denote by $\mathcal{S}.nz$ the number of  items in $\mathcal{S}$, and respectively with $\mathcal{S}[i].e$,  $\mathcal{S}[i].\hat{f}$ and $\mathcal{S}[i].\hat{\varepsilon}$ the element monitored by the $i$th counter of $\mathcal{S}$, the corresponding estimated frequency and the error committed in the estimation. 

The reduction works as follows. For both input summaries $\mathcal{S}_i, i = 1,2$ we have $\mathcal{S}_i.nz \leq k$.  We determine $m_1$ as the minimum among the frequencies of $\mathcal{S}_1$ if $\mathcal{S}_1.nz = k$, otherwise $m_1 =  0$. Similarly, we determine $m_2$ for $\mathcal{S}_2$.
Then, we combine the two summaries by calling the \textit{COMBINE} function, shown as pseudocode in Algorithm \ref{combine}. We scan the first hash table, and for each item in $\mathcal{S}_1$ we check if the item also appears in $\mathcal{S}_2$ by calling the \textit{FIND} function. In this case, we insert the entry for the item in $\mathcal{S}_C$, storing as its estimated frequency (respectively as its error) the sum of its frequency and the frequency of the corresponding item in $\mathcal{S}_2$ (respectively the sum of its error and the error of the corresponding item in $\mathcal{S}_2$), and remove the item from $\mathcal{S}_2$. Otherwise, we insert the entry for the item storing as its estimated frequency (respectively as its error) the sum of its frequency and the minimum $m_2$ (respectively the sum of its error and the minimum $m_2$). 

We then scan the second hash table. Since each time an item in $\mathcal{S}_1$ was also present in $\mathcal{S}_2$ we removed that item from $\mathcal{S}_2$, now $\mathcal{S}_2$ contains only items that do not appear in $\mathcal{S}_1$. For each item in $\mathcal{S}_2$ we simply insert the item in $\mathcal{S}_C$ and in the corresponding counter we store as estimated frequency (respectively as its error) the sum of its frequency and the minimum $m_1$ (respectively the sum of its error and the minimum $m_1$). Finally, the entries in $\mathcal{S}_C$ are sorted by the counters' frequency and this hash table is returned.

Note that for the $\mathcal{S}_C$ summary returned by the \textit{COMBINE} function it holds that $\mathcal{S}_C.nz \leq 2k$. Indeed, $\mathcal{S}_C$  may contain up to $2k$ items in the worst case (i.e., when all of the items in both $\mathcal{S}_1$ and $\mathcal{S}_2$ are different). 

However, we need to return at most $k$ items. Therefore, if $\mathcal{S}_C.nz \leq k$ (the number of entries with nonzero counter's frequency is at most $k$), we return $\mathcal{S}_C$ as $\mathcal{S}_M$. Otherwise, we remove the first $\mathcal{S}_C.nz - k$ items and then return $\mathcal{S}_C$ as $\mathcal{S}_M$, which contains exactly the $k$ items with the largest frequencies.

As an implementation detail, in the \textit{COMBINE} function it is possible to avoid using the $\mathcal{S}_C$ hash table altogether. Indeed, when scanning $\mathcal{S}_1$ one can simply update the frequency of the current item being processed, and when scanning $\mathcal{S}_2$ each item will be inserted into $\mathcal{S}_1$. At the end, $\mathcal{S}_1$ is returned. However, we prefer to use $\mathcal{S}_C$ in the interest of clarity, noting that the overall space complexity of \textit{COMBINE} is $O(k)$ in either case.

\begin{algorithm}
\begin{algorithmic}[1]
\Require $\mathcal{N}$, an array; $n$, the length of $\mathcal{N}$; $p$, the number of processors; $k$, the $k$-majority parameter
\Ensure an hash table containing $k$--majority candidate elements
\Procedure {ParallelSpaceSaving}{$\mathcal{N},  n, p, k$}
\Comment{The $n$ elements of the input array $\mathcal{N}$ are distributed to the $p$ processors so that each one is responsible for either $\left\lfloor {n/p} \right\rfloor$ or
$\left\lceil {n/p} \right\rceil$ elements; let $left$ and $right$ be respectively the indices of the first and last element of the sub-array handled by the process with rank
$id$; ranks are numbered from 0 to $p-1$}
\State $left \leftarrow \left\lfloor {(id-1)~n/p} \right\rfloor$
\State $right \leftarrow \left\lfloor {id~n/p} \right\rfloor  - 1$

\State $local \leftarrow  \Call{SpaceSaving}{\mathcal{N}, left, right, k}$
\Comment{determine local candidates}
\State let $hash$ be an hash table storing $<item, counter>$ pairs in $local$
\State sort $hash$ by counters' frequency in ascending order
\State $global \leftarrow \Call{ParallelReduction}{hash, k}$
\Comment{determine the global candidates for the whole array}

\If{$id == 0$} \Comment{we assume here that the processor with rank 0 holds the final result of the parallel reduction}
	\State $result \leftarrow \Call{Pruned}{global, n, k}$ 
	\State \Return $result$ 
\EndIf

\EndProcedure
\caption{Parallel space saving.}
\label{pss}
\end{algorithmic}
\end{algorithm}

\begin{algorithm}
\begin{algorithmic}[1]
\Require $\mathcal{S}_1$, $\mathcal{S}_2$: hash tables ordered by counters' frequency; $k$: the $k$-majority parameter; 
	the hash tables store pairs $<item, counter>$, a monitored item $e$ is used as key and a counter $c$ as object, including the estimated frequency $c.\hat{f}$ and the error $c.\hat{\varepsilon}$ of the item $e$
\Ensure an hash table, which is the \textit{merged summary} $\mathcal{S}_M$
\Procedure {ParallelReduction}{$\mathcal{S}_1,  \mathcal{S}_2, k$}
        \Comment $\mathcal{S}_i.nz$ is the number of items in the hash table $\mathcal{S}_i$
	\If{$\mathcal{S}_1.nz == k$}
		\State let $counter$ be the first counter in $\mathcal{S}_1$
		\State $m_1 \leftarrow counter.\hat{f}$
	\Else
		\State $m_1 \leftarrow 0$
	\EndIf
	\If{$\mathcal{S}_2.nz == k$}
		\State let $counter$ be the first counter in $\mathcal{S}_2$
		\State $m_2 \leftarrow counter.\hat{f}$
	\Else
		\State $m_2 \leftarrow 0$
	\EndIf
\State $\mathcal{S}_C \leftarrow \Call{combine}{\mathcal{S}_1, \mathcal{S}_2, m_1, m_2, k}$

\If{$\mathcal{S}_C.nz \leq k$}
	\State \Return $\mathcal{S}_C$ as $\mathcal{S}_M$;
\Else 
	\State $excess \leftarrow \mathcal{S}_C.nz - k$
	\State remove first $excess$ items from $\mathcal{S}_C$
	\State \Return $\mathcal{S}_C$ as $\mathcal{S}_M$; \Comment return the last $k$ items
\EndIf

\EndProcedure
\caption{Parallel reduction for space saving summaries.}
\label{pssr}
\end{algorithmic}
\end{algorithm}

\begin{algorithm}
\begin{algorithmic}[1]
\Require $\mathcal{S}_1$, $\mathcal{S}_2$: hash tables ordered by counters' frequency; $m_1$, the minimum of counters' frequency in $\mathcal{S}_1$; $m_2$, the minimum of counters' frequency in $\mathcal{S}_2$; $k$, the $k$-majority parameter
\Ensure an hash table, which is the \textit{combined summary} $\mathcal{S}_C$
\Procedure {combine}{$\mathcal{S}_1$,  $\mathcal{S}_2$, $m_1$, $m_2$, $k$}
\State let $\mathcal{S}_C$ be an empty hash table
\ForAll{$entry$ in $\mathcal{S}_1$}
	\State $item \leftarrow entry.key$
	\State $counter \leftarrow entry.val$	
	\State $found \leftarrow \mathcal{S}_2.\Call{Find}{item}$
	\If{$found$}
		\State $newcounter.\hat{f} \leftarrow counter.\hat{f} + found.\hat{f}$
		\State $newcounter.\hat{\varepsilon} \leftarrow counter.\hat{\varepsilon} + found.\hat{\varepsilon}$
		\State $\mathcal{S}_C.Put(item, newcounter)$
		\State $\mathcal{S}_2.Remove(item)$
	\Else
		\State $newcounter.\hat{f} \leftarrow counter.\hat{f} + min_2$
		\State $newcounter.\hat{\varepsilon} \leftarrow counter.\hat{\varepsilon} + min_2$
		\State $\mathcal{S}_C.Put(item, newcounter)$
	\EndIf
\EndFor

\ForAll{$entry$ in $\mathcal{S}_2$}
	\State $item \leftarrow entry.key$
	\State $counter \leftarrow entry.val$	
	\State $newcounter.\hat{f} \leftarrow counter.\hat{f} + min_1$
	\State $newcounter.\hat{\varepsilon} \leftarrow counter.\hat{\varepsilon} + min_1$
	\State $\mathcal{S}_C.Put(item, newcounter)$
\EndFor

\State sort $\mathcal{S}_C$ by counters' frequency in ascending order
\State \Return $\mathcal{S}_C$
\EndProcedure
\caption{Combine.}
\label{combine}
\end{algorithmic}
\end{algorithm}

\section{Correctness}
\label{correctness}

In this Section we formally prove that our Parallel Space Saving Algorithm is correct when executed on $p$ processors. We decompose the original array (i.e. multiset) of data $\mathcal{N}$ in $p$ subarrays $\mathcal{N}_i$ $(i=0,\ldots,p-1)$, namely $\mathcal{N}=\biguplus_i \mathcal{N}_i$. Here
the $\uplus$ operator denotes the \textit{join operation} \cite{Syropoulos01}, which is the sum of the frequency functions as follows:
$I_{\mathcal{A} \uplus \mathcal{B}} (x) = I_\mathcal{A} (x) + I_\mathcal{B} (x)$.
Let the sub--array $\mathcal{N}_i$ be assigned to the processor $p_i$, whose rank is denoted by $id$, with $id=0,\ldots, p-1$. Let also $\left\vert{\mathcal{N}_i}\right\vert$ denote the cardinality of $\mathcal{N}_i$, with $\sum_i \left\vert{\mathcal{N}_i}\right\vert=\left\vert{\mathcal{N}}\right\vert=n $.

The first step of the algorithm consists in the execution of the sequential Space Saving algorithm (which has already been proved to be correct by its authors),  on the subarray
assigned to each processor $p_{i}$. Therefore, in order to prove the overall correctness of the algorithm, we just need to demonstrate that the parallel reduction is correct.
Our strategy is to prove that if a single sub-step of the parallel reduction is correct, then we can naturally extend the proof to the $O(\log~p)$ steps of the whole parallel reduction.

We begin by proving a couple of preliminary results necessary for the proof of correctness of our parallel algorithm; both results are related to the combined summary $\mathcal{S}_C$ obtained by Algorithm \ref{combine}. We present in Table \ref{notation} the notation used throughout this Section, and recall here that we use a calligraphic capital letter to denote a multiset, and the corresponding capital Greek letter to denote its \textit{underlying} set.

\begin{table}
\renewcommand{\arraystretch}{1.3}
 \caption{Notation}
      \label{notation}
	\centering
    \begin{tabularx}{\textwidth}{@{} |c|X| @{}}
    \hline
    Notation & Description  \\ \hline \hline
    $\mathcal{L}$ &  A generic multiset (input or summary) \\ \hline
    $\Lambda$ & Underlying set of  $\mathcal{L}$  \\ \hline
    $\left\vert{\mathcal{L}}\right\vert$ & Cardinality of $\mathcal{L}$ \\ \hline
    $\left\vert{\Lambda}\right\vert$ & Cardinality of the underlying set of $\mathcal{L}$ \\ \hline
    $\hat{f}_{\mathcal{L}}(e)$ & Let $\mathcal{L}$ be a summary related to an input multiset $\mathcal{N}$; given an item $e \in \mathcal{L}$, $\hat{f}_{\mathcal{L}}(e)$ is the estimated frequency of item $e$ in $\mathcal{N}$ \\ \hline
    $f_{\mathcal{L}}(e)$ & Exact frequency of item $e$ in $\mathcal{L}$, an input multiset\\ \hline
    $\hat{\varepsilon}_{\mathcal{L}}(e)$ & Estimated 
    error of item $e$ in $\mathcal{L}$, a summary related to an input multiset \\ \hline
    $\hat{f}_{\mathcal{L}}^{min}$ & Minimum of counters' frequency in $\mathcal{L}$; we let $\hat{f}_{\mathcal{L}}^{min} = 0$ if $\left\vert{\Lambda}\right\vert < k$ \\ \hline
    \end{tabularx}
    \end{table}

Mathematically, we can express the combine operation as shown by the following two equations:

\begin{equation}
\label{eq08}
\hat{f}_{\mathcal{S}_C}(e) = 
\left\lbrace 
\begin{array}{r} 
\vspace{0.3cm}
\hat{f}_{\mathcal{S}_1}(e) + \hat{f}_{\mathcal{S}_2}(e), \qquad e \in \Sigma_1 \cap \Sigma_2, \\
\vspace{0.3cm}
\hat{f}_{\mathcal{S}_1}(e) + \hat{f}_{\mathcal{S}_2}^{min}, \qquad e \in \Sigma_1 \setminus \Sigma_2, \\
\hat{f}_{\mathcal{S}_2}(e) + \hat{f}_{\mathcal{S}_1}^{min}, \qquad e \in \Sigma_2 \setminus \Sigma_1,
\end{array}\right.
\end{equation}

\begin{equation}
\label{eq09}
\hat{\varepsilon}_{S_C}(e) = 
\left\lbrace 
\begin{array}{r} 
\vspace{0.3cm}
\hat{\varepsilon}_{\mathcal{S}_1}(e) + \hat{\varepsilon}_{\mathcal{S}_2}(e), \qquad e \in \Sigma_1 \cap \Sigma_2, \\
\vspace{0.3cm}
\hat{\varepsilon}_{\mathcal{S}_1}(e) + \hat{f}_{\mathcal{S}_2}^{min}, \qquad e \in \Sigma_1 \setminus \Sigma_2, \\
\hat{\varepsilon}_{\mathcal{S}_2}(e) + \hat{f}_{\mathcal{S}_1}^{min}, \qquad e \in \Sigma_2 \setminus \Sigma_1.
\end{array}\right.
\end{equation}

As shown in eq. (\ref{eq08}), if an item belongs to both summaries, we update its estimated frequency by summing up the estimated frequencies of the counters corresponding to the item in the two summaries. If an item belongs to only one summary, we update its estimated frequency by adding the minimum frequency stored in the other summary. At the same time we can estimate the error for each item, as shown by eq. (\ref{eq09}). This combining step leads to a summary $\mathcal{S}_C$ storing at most $2k$ distinct items. This summary includes all of the frequent items belonging to the set of items which is the union of the underlying sets related to the two input summaries.

\begin{lem}
\label{lemma1}
Let $\mathcal{S}_1 = (\Sigma_1, \hat{f}_{\mathcal{S}_1})$ and $\mathcal{S}_2 = (\Sigma_2, \hat{f}_{\mathcal{S}_2})$ be two summaries related respectively to the input sub-arrays $\mathcal{N}_1 = (N_1, f_{\mathcal{N}_1})$ and $\mathcal{N}_2 = (N_2, f_{\mathcal{N}_2})$, with $\mathcal{N} = \mathcal{N}_1 \uplus \mathcal{N}_2 = (N, f_{\mathcal{N}})$. Let $\mathcal{S}_C = (\Sigma_C, \hat{f}_{\mathcal{S}_C})$ be the intermediate summary obtained combining $\mathcal{S}_1$ and $\mathcal{S}_2$, and let $\delta = \hat{f}_{\mathcal{S}_1}^{min} + \hat{f}_{\mathcal{S}_2}^{min}$.


\noi The following relation holds:

\begin{equation}
\label{eq10}
\left\vert{\mathcal{S}_C}\right\vert = \left\vert{\mathcal{S}_1}\right\vert + \left\vert{\mathcal{S}_2}\right\vert + x \delta,
\end{equation}
\end{lem} 

\noi where $x = \left\vert{\Sigma_C}\right\vert - k$.

\begin{proof}
Let $c = \left\vert{\Sigma_1 \cap \Sigma_2}\right\vert$, $d_1 = \left\vert{\Sigma_1 \setminus \Sigma_2}\right\vert $ and  $d_2 = \left\vert{\Sigma_2 \setminus \Sigma_1}\right\vert$. Then, $x = c + d_1 + d_2 - k$. It follows that 

\begin{equation}
\label{eq11}
\begin{aligned}
x \delta&=(c + d_1 + d_2 - k) \delta \\
&=(c + d_1) \hat{f}_{S_1}^{min} + (c + d_2) \hat{f}_{S_2}^{min}  -  \ k \hat{f}_{S_1}^{min} - k \hat{f}_{S_2}^{min} + d_1 \hat{f}_{\mathcal{S}_2}^{min} + d_2 \hat{f}_{\mathcal{S}_1}^{min}.
\end{aligned}
\end{equation}

\noi Since $\left\vert{\Sigma_1}\right\vert = c + d_1 \leq k$ and $\left\vert{\Sigma_2}\right\vert = c + d_2 \leq k$, and observing that $\left\vert{\Sigma_1}\right\vert < k \Leftrightarrow \hat{f}_{\mathcal{S}_1}^{min} = 0$ and $\left\vert{\Sigma_2}\right\vert < k \Leftrightarrow \hat{f}_{\mathcal{S}_2}^{min} = 0$, it follows that eq. (\ref{eq11}) reduces to

\begin{equation}
\label{eq12}
x \delta = d_1 \hat{f}_{\mathcal{S}_2}^{min} + d_2 \hat{f}_{\mathcal{S}_1}^{min}.
\end{equation}

\noi Therefore, we can rewrite eq. (\ref{eq10}) as

\begin{equation}
\label{eq13}
\left\vert{\mathcal{S}_C}\right\vert = \left\vert{\mathcal{S}_1}\right\vert + \left\vert{\mathcal{S}_2}\right\vert +  d_1 \hat{f}_{\mathcal{S}_2}^{min} + d_2 \hat{f}_{\mathcal{S}_1}^{min}.
\end{equation}

\noi This equation expresses the fact that the sum of the frequencies stored in $\mathcal{S}_C$ can be computed according to the way we combine the summaries in eq. (\ref{eq08}). Precisely, if $d_1 = 0$ and $d_2 = 0$, then $\mathcal{S}_1$ and $\mathcal{S}_2$ share all of the elements, so that

\begin{equation}
\label{eq14}
\left\vert{\mathcal{S}_C}\right\vert = \left\vert{\mathcal{S}_1}\right\vert + \left\vert{\mathcal{S}_2}\right\vert.
\end{equation}

\noi Otherwise, for items belonging to just one of the summaries, we add to their frequencies the minimum frequency of the other summary. In other words, besides their frequency (which is taken into account by $\left\vert{\mathcal{S}_1}\right\vert + \left\vert{\mathcal{S}_2}\right\vert$) we also add exactly $d_1 \hat{f}_{\mathcal{S}_2}^{min} + d_2 \hat{f}_{\mathcal{S}_1}^{min}$.

\end{proof}

\begin{lem}
\label{lemma2}
Let the summaries $\mathcal{S}_1$, $\mathcal{S}_2$ and $\mathcal{S}_C$, the input multisets $\mathcal{N}_1$ and $\mathcal{N}_2$ and the quantity $\delta$ be defined as in Lemma \ref{lemma1}. Assume that the following inequalities hold for each item $e \in N_1$: 

\begin{equation}
\label{eq15}
\hat{f}_{\mathcal{S}_1}(e) - \hat{f_\mathcal{S}}_1^{min} \leq \hat{f}_{\mathcal{S}_1}(e) - \hat{\varepsilon}_{\mathcal{S}_1}(e) \leq f_{\mathcal{N}_1}(e) \leq \hat{f}_{\mathcal{S}_1}(e), \qquad e \in \Sigma_1,
\end{equation}

\begin{equation}
\label{eq16}
f_{\mathcal{N}_1}(e)  \leq \hat{f_\mathcal{S}}_1^{min}, \qquad \hspace{50 mm} e \notin \Sigma_1.
\end{equation}

\noi Similarly, assume that the following inequalities hold for each item $e \in N_2$: 

\begin{equation}
\label{eq17}
\hat{f}_{\mathcal{S}_2}(e) - \hat{f_\mathcal{S}}_2^{min} \leq \hat{f}_{\mathcal{S}_2}(e) - \hat{\varepsilon}_{\mathcal{S}_2}(e) \leq f_{\mathcal{N}_2}(e) \leq \hat{f}_{\mathcal{S}_2}(e), \qquad  e \in \Sigma_2,
\end{equation}

\begin{equation}
\label{eq18}
f_{\mathcal{N}_2}(e)  \leq \hat{f_\mathcal{S}}_2^{min}, \qquad \hspace{50mm} e \notin \Sigma_2.
\end{equation}

\noi Then, for each item $e \in N$ we have:

\begin{equation}
\label{eq19}
\hat{f}_{\mathcal{S}_C}(e) - \delta \leq  \hat{f}_{\mathcal{S}_C}(e) -  \hat{\varepsilon}_{\mathcal{S}_C}(e) \leq    f_\mathcal{N}(e) \leq  \hat{f}_{\mathcal{S}_C}(e), \qquad  \hspace{6 mm} e \in \Sigma_C,
\end{equation}

\begin{equation}
\label{eq20}
f_\mathcal{N}(e) \leq \delta, \qquad \hspace{57 mm} e \notin \Sigma_C.
\end{equation}
 
\end{lem}

\begin{proof}

The summary $\mathcal{S}_C$ is derived from $\mathcal{S}_1$ and $\mathcal{S}_2$ by applying eqs. (\ref{eq08}) and (\ref{eq09}), so that, in order to prove eq. (\ref{eq19}) we need to distinguish three cases:

\begin{enumerate}
  \item Let $e \in \Sigma_1 \cap \Sigma_2$: $\hat{f}_{\mathcal{S}_1}(e) + \hat{f}_{\mathcal{S}_2}(e) - (\hat{f}_{\mathcal{S}_1}^{min} + \hat{f}_{\mathcal{S}_2}^{min}) \leq \hat{f}_{\mathcal{S}_1}(e) + \hat{f}_{\mathcal{S}_2}(e) - (\hat{\varepsilon}_{\mathcal{S}_1}(e) + \hat{\varepsilon}_{\mathcal{S}_2}(e)) \leq f_{\mathcal{N}_1}(e) + f_{\mathcal{N}_2}(e) \leq \hat{f}_{\mathcal{S}_1}(e) + \hat{f}_{\mathcal{S}_2}(e)$. But, by definition in this case it is $\hat{f}_{\mathcal{S}_C}(e) = \hat{f}_{\mathcal{S}_1}(e) + \hat{f}_{\mathcal{S}_2}(e)$, $\hat{\varepsilon}_{\mathcal{S}_C}(e) = \hat{\varepsilon}_{\mathcal{S}_1}(e) + \hat{\varepsilon}_{\mathcal{S}_2}(e)$, 
       $f_{\mathcal{N}}(e) = f_{\mathcal{N}_1}(e) + f_{\mathcal{N}_2}(e)$ and $\delta = \hat{f}_{\mathcal{S}_1}^{min} + \hat{f}_{\mathcal{S}_2}^{min}$, so that, taking into account eqs. (\ref{eq15}) and (\ref{eq17}), eq. (\ref{eq19}) holds;
  
  \item Let $e \in \Sigma_1 \setminus \Sigma_2$: following the same reasoning as before, taking into account eqs. (\ref{eq15}) and (\ref{eq18}) and that in this case it is by definition $\hat{f}_{\mathcal{S}_C}(e) = \hat{f}_{\mathcal{S}_1}(e) + \hat{f}_{\mathcal{S}_2}^{min}$ and $\hat{\varepsilon}_{\mathcal{S}_C}(e) = \hat{\varepsilon}_{\mathcal{S}_1}(e) + \hat{f}_{\mathcal{S}_2}^{min}$, we obtain $\hat{f}_{\mathcal{S}_1}(e) - \hat{f}_{\mathcal{S}_1}^{min} \leq \hat{f}_{\mathcal{S}_1}(e) - \hat{\varepsilon}_{\mathcal{S}_1}(e) \leq f_{\mathcal{N}_1}(e) + f_{\mathcal{N}_2}(e) \leq \hat{f}_{\mathcal{S}_1}(e) + \hat{f}_{\mathcal{S}_2}^{min}$, so that $\hat{f}_{\mathcal{S}_1}(e) - \hat{f}_{\mathcal{S}_1}^{min} \leq \hat{f}_{\mathcal{S}_1}(e) - \hat{\varepsilon}_{\mathcal{S}_1}(e) \leq f_{\mathcal{N}}(e) \leq \hat{f}_{\mathcal{S}_C}(e)$. Rewriting $\hat{f}_{\mathcal{S}_1}(e) - \hat{f}_{\mathcal{S}_1}^{min}$ as $\hat{f}_{\mathcal{S}_1}(e) + \hat{f}_{\mathcal{S}_2}^{min} - \hat{f}_{\mathcal{S}_2}^{min} - \hat{f}_{\mathcal{S}_1}^{min}$, and $\hat{f}_{\mathcal{S}_1}(e) - \hat{\varepsilon}_{\mathcal{S}_1}(e)$ as $\hat{f}_{\mathcal{S}_1}(e) + \hat{f}_{\mathcal{S}_2}^{min} - \hat{f}_{\mathcal{S}_2}^{min} - \hat{\varepsilon}_{\mathcal{S}_1}(e)$ we obtain eq. (\ref{eq19});

  \item Let $e \in \Sigma_2 \setminus \Sigma_1$: immediate, taking into account eqs. (\ref{eq16}) and (\ref{eq17}), since this case is symmetric to the previous one.
\end{enumerate}

To prove eq. (\ref{eq20}), taking into account eqs. (\ref{eq16}) and (\ref{eq18}) we obtain for an item $e \notin \Sigma_C$: $f_{\mathcal{N}_1}(e) + f_{\mathcal{N}_2}(e) \leq \hat{f}_{\mathcal{S}_1}^{min} + \hat{f}_{\mathcal{S}_2}^{min}$, i.e., $f_{\mathcal{N}}(e) \leq \delta$.

\end{proof}

Now we can formally prove the correctness of our parallel algorithm. Let us consider how it works. Before engaging in the parallel reduction, each processor applies the sequential Space Saving algorithm to its local input, producing an hash table data structure containing at most $k$ counters with estimated frequency greater than zero. In the parallel reduction, we merge pairs of data structures until we output the final result. 

We start by proving the following

\begin{thm}
\label{pr-single-step}
A single reduction sub-step correctly merges its two input summaries.
\end{thm}

\begin{proof}
Let $\mathcal{S}_1 = (\Sigma_1, \hat{f}_{\mathcal{S}_1})$ and $\mathcal{S}_2 = (\Sigma_2, \hat{f}_{\mathcal{S}_2})$ be two summaries related respectively to the input sub-arrays $\mathcal{N}_1 = (N_1, f_{\mathcal{N}_1})$ and $\mathcal{N}_2 = (N_2, f_{\mathcal{N}_2})$, with $\mathcal{N} = \mathcal{N}_1 \uplus \mathcal{N}_2 = (N, f_{\mathcal{N}})$. Let $\mathcal{S}_C = (\Sigma_C, \hat{f}_{\mathcal{S}_C})$ be the intermediate summary obtained combining $\mathcal{S}_1$ and $\mathcal{S}_2$ by using the \textit{COMBINE} function, and let $\mathcal{S}_M = (\Sigma_M, \hat{f}_{\mathcal{S}_M})$ be the final merged summary. 

We are going to prove that if eqs. (\ref{eq05}) - (\ref{eq07}) hold for $\mathcal{S}_1$ and $\mathcal{S}_2$ and, if it is verified a relaxed version of eq. (\ref{eq04}), i.e., for a summary $\mathcal{S}$ it holds that

\begin{equation}
\label{eq21}
\left\vert{\mathcal{S}}\right\vert \leq \left\vert{\mathcal{N}}\right\vert,
\end{equation}

\noi then these properties  continue to be true also for $\mathcal{S}_M$ (it is worth noting here that eq. (\ref{eq21}) also holds for summaries produced by the sequential Space Saving algorithm). We shall show that this is enough to guarantee the correctness of the merge operation. 
 
The \textit{merge} operation is done in two steps and provides as output a summary of at most $k$ items. In the first step we combine the input summaries as shown in eqs. (\ref{eq08}) and (\ref{eq09}). This combining step leads to an intermediate summary $\mathcal{S}_C$ storing at most $2k$ distinct items.

In the second and final step, we analyze $\mathcal{S}_C$ in order to return the final output. If $\mathcal{S}_C$ holds at most $k$ entries (i.e., $\left\vert{\Sigma_C}\right\vert \leq k$), we return $\mathcal{S}_C$ as the output, i.e., $\mathcal{S}_M = \mathcal{S}_C$. However, if $\mathcal{S}_C$ holds more than $k$ entries (the data structure may hold in the worst case up to $2k$ entries), we need to select and return $k$ entries. In this case, we simply return as $\mathcal{S}_M$ the last $k$ entries in $\mathcal{S}_C$, i.e., those corresponding to the items with greatest frequency (the entries are sorted by counters' frequency). 

We start by noting that in the summary $\mathcal{S}_C$ generated by the first step a counter's frequency still represents an excess estimation of the monitored item, as in Space Saving. As before, let $\delta = \hat{f}_{S_1}^{min} +\hat{f}_{S_2}^{min}$, and $x = \left\vert{\Sigma_C}\right\vert - k$.

%
%
%
%

 

By Lemma \ref{lemma1}, eq. (\ref{eq10}), if $\left\vert{\Sigma_C}\right\vert \leq k$, then $x \delta = 0$ (indeed, when $\left\vert{\Sigma_C}\right\vert < k$  then $\delta = 0$, when $\left\vert{\Sigma_C}\right\vert = k$ then $x = 0$) and the merged summary $\mathcal{S}_M$ coincides with $\mathcal{S}_C$. In that case, since by eq. (\ref{eq21}) $\left\vert\mathcal{S}_1\right\vert \leq \left\vert\mathcal{N}_1\right\vert$ and $\left\vert\mathcal{S}_2\right\vert \leq \left\vert\mathcal{N}_2\right\vert$, we have that $\left\vert{\mathcal{S}_M}\right\vert = \left\vert{\mathcal{S}_1}\right\vert + \left\vert{\mathcal{S}_2}\right\vert \leq \left\vert{\mathcal{N}}\right\vert$, so that eq. (\ref{eq21}) also holds for $\mathcal{S}_M$. Otherwise, if $\left\vert{\Sigma_C}\right\vert > k$, in order to obtain the final merged summary, we return in $\mathcal{S}_M$ the $k$ items in $\mathcal{S}_C$ with the highest frequencies. Precisely, let the entries in $\mathcal{S}_C$ be sorted in ascending order with regard to the counters' frequencies. Then,

\begin{equation}
\label{eq22}
\left\vert{\mathcal{S}_M}\right\vert = \left\vert{\mathcal{S}_1}\right\vert + \left\vert{\mathcal{S}_2}\right\vert + x \delta - \sum_{i=1}^x \hat{f}_{\mathcal{S}_C}(e_i),
\end{equation}

\noi where the sum is extended over the first $x$ entries. We observe that $\sum_{i=1}^x \hat{f}_{S_C}(e_i) \geq x \delta$, owing to the fact that the counters appear in sorted order, and the estimated frequencies stored in each of the initial $x$ counters we are discarding are greater than or equal to $\delta$. In this case too, it follows that eq. (\ref{eq21}) holds for $\mathcal{S}_M$. Indeed,

\begin{equation}
\label{eq23}
\left\vert{\mathcal{S}_M}\right\vert \leq \left\vert{\mathcal{S}_1}\right\vert + \left\vert{\mathcal{S}_2}\right\vert \leq \left\vert{\mathcal{N}_1}\right\vert + \left\vert{\mathcal{N}_2}\right\vert = \left\vert{\mathcal{N}}\right\vert.
\end{equation}

We have to prove that the other properties are verified as well. In particular we have to show that the error bound guaranteed by the sequential Space Saving algorithm is preserved by the merge operation. In order to do this, we observe that $\hat{f}_{\mathcal{S}_M}^{min}$ is such that (see the similar proof of Lemma 3.3 in \cite{Metwally2006})

\begin{equation}
\label{eq24}
\hat{f}_{\mathcal{S}_M}^{min} =  \frac{\left\vert{\mathcal{S}_M}\right\vert - \sum_{e \in \Sigma_M} (\hat{f}_{\mathcal{S}_M}(e) - \hat{f}_{\mathcal{S}_M}^{min})}{k}.
\end{equation}

At the same time, $\sum_{e \in \Sigma_M} (\hat{f}_{\mathcal{S}_M}(e) - \hat{f}_{\mathcal{S}_M}^{min}) \geq 0$, because the frequency of each item is greater than or equal to the minimum. Therefore we have:

\begin{equation}
\label{eq25}
\hat{f}_{\mathcal{S}_M}^{min} \leq  \frac{\left\vert{\mathcal{S}_M}\right\vert}{k}.
\end{equation}

Observing that $\hat{f}_{\mathcal{S}_M}^{min} \geq \delta$ and taking into account eq. (\ref{eq25}) and the fact that eq. (\ref{eq21}) also holds for $\mathcal{S}_M$, we can bound $\hat{f}_{\mathcal{S}_M}^{min}$ as follows:

\begin{equation}
\label{eq26}
\delta \leq \hat{f}_{\mathcal{S}_M}^{min} \leq  \frac{\left\vert{\mathcal{S}_M}\right\vert}{k} \leq \left\lfloor\frac{\left\vert{\mathcal{N}}\right\vert}{k}\right\rfloor.
\end{equation}

At last, taking into account Lemma \ref{lemma2}, eqs. (\ref{eq19}) and (\ref{eq20}) and the way we construct $\mathcal{S}_M$, we have that, for each item $e \in N$ (i.e., for each distinct item in the input $\mathcal{N} = \mathcal{N}_1 \uplus \mathcal{N}_2$):

\begin{equation}
\label{eq27}
\hat{f}_{\mathcal{S}_M}(e) - \hat{f}_{\mathcal{S}_M}^{min} \leq \hat{f}_{\mathcal{S}_M}(e) - \hat{\varepsilon}_{\mathcal{S}_M}(e) \leq f_\mathcal{N}(e) \leq  \hat{f}_{\mathcal{S}_M}(e), \qquad e \in \Sigma_M,
\end{equation}

and

\begin{equation}
\label{eq28}
f_\mathcal{N}(e) \leq  \hat{f}_{\mathcal{S}_M}^{min} \leq \left\lfloor\frac{\left\vert{\mathcal{N}}\right\vert}{k}\right\rfloor, \qquad \hspace{40 mm} e \notin \Sigma_M,
\end{equation}

\noi showing that eqs. (\ref{eq05}) - (\ref{eq07}) also hold for $\mathcal{S}_M$.

\end{proof}

\noi It is worth noting here that a single reduction step (i.e., a parallel execution with $p = 2$ processors) is fully equivalent to a sequential algorithm for merging two data summaries. Therefore, Theorem \ref{pr-single-step} states the correctness of this algorithm.  We can now prove the following

\begin{prop}
\label{pr-whole}
The whole parallel reduction correctly merges its input summaries.
\end{prop}

\begin{proof}

The correctness of the whole parallel reduction follows straightforwardly. Indeed, it is enough noting that in the initial step of the parallel reduction we process summaries derived by applying locally in each processor Space Saving, and eqs. (\ref{eq05}) - (\ref{eq07}) and (\ref{eq21}) hold for these summaries. By Theorem \ref{pr-single-step}, the merge operation used in each sub-step of the parallel reduction outputs a summary for which eqs. (\ref{eq05}) - (\ref{eq07}) and (\ref{eq21}) continue to hold and whose error is still within the bound on the error committed estimating the frequencies guaranteed by Space Saving. Therefore, at the end of the $O(\log~p)$ steps required for the whole reduction, the final output summary correctly provides the frequent items for the whole input.

\end{proof}

The main result of this Section is the following 

\begin{thm} Algorithm \ref{pss} correctly determines frequent items in parallel.
\end{thm}

\begin{proof}
The result follows immediately from Theorem \ref{pr-single-step} and Proposition \ref{pr-whole}.
\end{proof}

\section{Parallel complexity}
\label{pss-analysis}
In this Section, we discuss the parallel complexity of the proposed parallel Space Saving algorithm. We assume, in the following analysis, that $k = O(1)$. The assumption is justified by the fact that it is verified in all of the cases of practical interest for this application.

At the beginning of the algorithm, the workload is balanced using a block
distribution; this is done with two simple $O(1)$ assignments; therefore, the complexity of the initial domain decomposition is $O(1)$. Next,
we determine local candidates in each subarray using the sequential Space Saving algorithm. Owing to the block distribution and to the fact that Space Saving is linear in the
number of input elements, the complexity of this step is $O(n/p)$. Then, we engage in a parallel reduction to
determine the global candidates for the whole input array. The whole reduction requires in the worst case $O(log~p)$.

Indeed, in the initial step we combine the input summaries in $O(k)$, by using hash tables. For each item in the $\mathcal{S}_1$ hash table, we try to find in $O(1)$ time a corresponding item in the $\mathcal{S}_2$ hash table. Then, we insert in the $\mathcal{S}_C$ hash table the entry for the item, again in $O(1)$ time and, if we have found the item, we delete the corresponding entry from $\mathcal{S}_2$ in $O(1)$ time. Since there are at most $k$ entries, this requires $O(k)$. We then scan the entries in $\mathcal{S}_2$  (there can be at most $k$ entries, this happens when the items in the two hash tables are all distinct, otherwise there will be less than $k$ entries because we remove corresponding items from $\mathcal{S}_2$ each time we find a match). For each entry in $\mathcal{S}_2$, we simply insert the corresponding item in $\mathcal{S}_C$ in $O(1)$ time. Therefore, processing $\mathcal{S}_2$ requires in the worst case $O(k)$ time.

In the second step, we simply return the combined summaries if the total number of entries in $\mathcal{S}_C$ is less than or equal to $k$, otherwise, we return the last $k$ entries in sorted order of $\mathcal{S}_C$. The time required is $O(k)$.

To recap, since we do $O(k)$ work in each step of the parallel reduction, $k = O(1)$ by assumption and there are $O(\log~p)$ such steps, the overall complexity of the reduction is  $O(log~p)$. The communication cost, i.e., the amount of data exchanged in the parallel reduction is 

\begin{equation}
\left\lceil {\sum\limits_{i = 1}^{\log p} {\frac{p}{{{2^i}}}} k} \right\rceil  = (p - 1)k = O(pk) = O(p),
\end{equation}

\noi since $k = O(1)$ by assumption.

Finally, the worst case complexity of the $Pruned$ function is $O(k) = O(1)$, since this is just a linear scan in which we compare the frequency of each item against the threshold required to be a frequent item, and put an item in the $result$  summary if its frequency is greater than or equal to the required threshold.

It follows that the overall  complexity of the parallel Space Saving algorithm is $O(n/p + \log~p)$. We are now in the position to state the following Theorem:

\begin{thm}
The algorithm is cost--optimal for $k = O(1)$.
\end{thm}

\begin{proof}
Cost--optimality requires by definition that asymptotically $pT_p = T_1$ where $T_1$  represents the time spent on one processor (sequential time) and $T_p$ the time spent on
$p$ processors. The sequential algorithm requires $O(n)$ in the worst case, and the parallel complexity of our algorithm is $O(n/p + \log~p)$ when $k = O(1)$. It follows from
the definition that the algorithm is cost--optimal for $n = \Omega(p \log p)$.
\end{proof}

Cost--optimality is an important theoretical property of parallel algorithms, since it implies linear speedup (equal to $p$) and efficiency equal to 1. Moreover,  cost--optimality also implies good scalability of the algorithm when using smaller sized parallel machines equipped with a limited number of processors. Indeed, scaling down a cost--optimal algorithm on a reduced number of processors will result in a fast algorithm, while scaling down a non cost--optimal algorithm may even result in a parallel algorithm doing more work and being slower than the corresponding best sequential algorithm. 

We proceed with the analysis of isoefficiency and scalability. The sequential algorithm has complexity $O(n)$; the parallel overhead is $T_o = pT_p - T_1$. In our case,  $T_o =
p(n/p + \log~p) - n =  p\log~p$.  The isoefficiency relation \cite{Grama93} is then $n \geq  p~\log~p$. Finally, we derive the scalability function of this parallel system
\cite{Quinn03}.

This function shows how memory usage per processor must grow to maintain efficiency at a desired level. If the isoefficiency relation is $n \geq f(p)$ and $M(n)$ denotes the
amount of memory required for a problem of size $n$, then $M (f (p))/p$ shows how memory usage per processor must increase to maintain the same level of efficiency. Indeed, in
order to maintain efficiency when increasing $p$, we must increase $n$ as well, but on parallel computers the maximum problem size is limited by the available memory, which is
linear in $p$. Therefore, when the scalability function $M (f (p))/p$ is a constant $C$, the parallel algorithm is perfectly scalable; $Cp$ represents instead the limit for
scalable algorithms. Beyond this point an algorithm is not scalable (from this point of view).
In our case the function describing how much memory is used for a problem of size $n$ is given by $M (n) = n$. Therefore, $M (f (p))/p = O(\log~p)$ with $f(p)$ given by the
isoefficiency relation.

\section{Conclusions}
\label{conclusions}

To the best of our knowledge, we have designed and implemented the first message-passing based parallel version of the Space Saving algorithm to solve the $k$--majority problem. In particular, we have shown that our algorithm retains all of the key features of the sequential Space Saving algorithm. Besides proving its formal correctness, we have applied our algorithm to the detection of frequent items in real datasets and in synthetic datasets whose probability distribution functions are a Hurwitz and a Zipf distribution respectively. Extensive experimental results on both synthetic and real datasets have been reported and discussed in Appendix, clearly showing that our algorithm outperforms the parallel version of the algorithm designed by Agarwal et al. with regard to precision, total error and average relative error, while providing overall comparable parallel performances with linear speedup.

\section*{Acknowledgment}
We are indebted to the unknown referees for enlightening observations, which helped us to improve the paper. The authors would also like to thank G. Cormode and M. Hadjieleftheriou for making freely available their sequential implementation of the Space Saving algorithm. We are also grateful to Prof. Palpanas of Paris Descartes University for providing us with the real datasets used in the experiments. The research of M. Cafaro has been supported by CMCC, Italy, under the grant FISR Gemina project, Italian Ministry of Education, University and Research. The research of P. Tempesta has been supported by the grant FIS2011--22566, Ministerio de Ciencia e Innovaci\'on, Spain.

\appendix
\section{Experimental results}
\label{appendix}

We report here the experimental results we have obtained running the parallel Space Saving algorithm on an IBM iDataPlex cluster. Each SMP node is configured with two 2.6 Ghz octa-core Xeon Sandy Bridge E5-2670 CPUs  with 20 MB level 3 cache and 64 GB of main memory. The interconnection network is  Infiniband 4x FDR-10 (Fourteen Data Rate) 40 Gbp/s, which provides 5 GB/s unidirectional bandwidth. Our parallel implementation, developed in C++ using MPI,  is based on the sequential source code for the \emph{Space Saving} algorithm developed in \cite{Cormode-code}.

In order to assess the merits of our parallel algorithm, we also compare it with a second parallel algorithm which we have designed and implemented starting from a sequential algorithm by Agarwal et al \cite{Agarwal}. The authors designed their algorithm for merging Frequent summaries, and then proved that for Space Saving summaries subtracting the minima from their respective summaries (if a summary possesses $k$ counters) makes them isomorphic to Frequent summaries, so that their algorithm can be reused (see Lemma 2 in \cite{Agarwal}).

In the \emph{ParallelAgarwal} algorithm (see Algorithm \ref{pa}) each processor starts by executing the Space Saving algorithm on its local sub-array. Then, just before engaging in the parallel reduction, if the \emph{local} summary holds $k$ nonzero counters, the minimum frequency, which is stored in the first counter \emph{local}[1], is subtracted from each counter. It follows that the \emph{local} summary stores at most $k-1$ counters, so that the algorithm by Agarwal et al. shown as the \emph{AgarwalParallelReduction} (see Algorithm \ref{par}), can be applied. The input of the parallel reduction is an hash table, storing the entries in $local$ sorted by counters' frequency.

\begin{algorithm}
\begin{algorithmic}[1]
\Require $\mathcal{N}$, an array; $n$, the length of $\mathcal{N}$; $p$, the number of processors; $k$, the $k$-majority parameter
\Ensure an hash table containing $k$--majority candidate elements
\Procedure {ParallelAgarwal}{$\mathcal{N},  n, p, k$}
\Comment{The $n$ elements of the input array $\mathcal{N}$ are distributed to the $p$ processors so that each one is responsible for either $\left\lfloor {n/p} \right\rfloor$ or
$\left\lceil {n/p} \right\rceil$ elements; let $left$ and $right$ be respectively the indices of the first and last element of the sub-array handled by the process with rank
$id$; ranks are numbered from 0 to $p-1$}
\State $left \leftarrow \left\lfloor {(id-1)~n/p} \right\rfloor$
\State $right \leftarrow \left\lfloor {id~n/p} \right\rfloor  - 1$
\State $local \leftarrow \Call{SpaceSaving}{\mathcal{N}, left, right}$
\Comment{determine local candidates}

	\If{$local.nz == k$} \Comment $local.nz$ is the number of items in the stream summary $local$ with nonzero frequency
		\State $m_1 \leftarrow \mathop {local[1].\hat{f}}$
		\For{$i=1$ to $k$}
			\State $\mathop {local[i].\hat{f}} \leftarrow \mathop {local[i].\hat{f}} - m_1$
		\EndFor
	\EndIf
\State let $hash$ be an hash table storing $<item, counter>$ pairs in $local$
\State sort $hash$ by counters' frequency in ascending order
\State $global \leftarrow \Call{AgarwalParallelReduction}{hash, k}$
\Comment{determine the global candidates for the whole array}
\If{$id == 0$} \Comment{we assume here that the processor with rank 0 contains the final result of the parallel reduction}
	\State \Return $global$
\EndIf
\EndProcedure
\caption{Parallel algorithm by Agarwal et al.}
\label{pa}
\end{algorithmic}
\end{algorithm}

Although the algorithm is presented in the context of merging two summaries, it can actually be used in parallel as a reduction operator, owing to the fact that the authors also proved a bound on the output error, which is within the error affecting the input summaries.

\begin{algorithm}
\begin{algorithmic}[1]
\Require $\mathcal{S}_1$, $\mathcal{S}_2$: hash tables; $k$, $k$-majority parameter (the number of counters is at most $k-1$);
\Ensure an hash table containing $k$--majority candidate elements
\Procedure{AgarwalParallelReduction}{$\mathcal{S}_1,  \mathcal{S}_2, k$} \Comment{a merged summary of $\mathcal{S}_1$ and $\mathcal{S}_2$}
\State $\mathcal{S} \leftarrow \Call{agarwal-combine}{\mathcal{S}_1, \mathcal{S}_2};$
\State sort $\mathcal{S}$ by counters' frequency in ascending order
\If{$\mathcal{S}.nz \leq k-1$}
	\State \Return $\mathcal{S}$;
\Else \Comment{prune counters in $\mathcal{S}$}

	\State $excess \leftarrow \mathcal{S}.nz - k + 1$
	\Comment{determine frequency to be subtracted}
	\State $entry = \mathcal{S}[excess]$
		\State $counter \leftarrow entry.val$	
		\State $freq \leftarrow counter.\hat{f}$
	\Comment{subtract this frequency from the last $k-1$ counters}
	\For{$i=excess+1$ to $\mathcal{S}.nz$}
		\State $entry = \mathcal{S}[i]$
		\State $item \leftarrow entry.key$
		\State $counter \leftarrow entry.val$	
		\State $frequency \leftarrow counter.\hat{f}$
		\State $\mathcal{S}.Update(item, frequency - freq)$
	\EndFor
	\State remove first $excess$ items from $\mathcal{S}$
	\State \Return $\mathcal{S}$;
\EndIf
\EndProcedure
\caption{Parallel Reduction by Agarwal et al.}
\label{par}
\end{algorithmic}
\end{algorithm}

The parallel reduction works as follows. It starts combining the two data sets, by calling the \textit{AGARWAL-COMBINE} function. Let $\mathcal{S}$ be the combined summary. Scanning the first hash table, for each item in $\mathcal{S}_1$ the function checks if the item also appears in $\mathcal{S}_2$. In this case, it inserts the entry for the item in $\mathcal{S}$, storing as its estimated frequency the sum of the item's frequency and the frequency of the corresponding item in $\mathcal{S}_2$, and removes the item from $\mathcal{S}_2$. Otherwise, the function inserts the entry for the item storing as its estimated frequency its frequency in $\mathcal{S}_1$. 

The function then scans the second hash table. Since each time an item in $\mathcal{S}_1$ was also present in $\mathcal{S}_2$ it was removed from $\mathcal{S}_2$, now $\mathcal{S}_2$ contains only items that do not appear in $\mathcal{S}_1$. For each item in $\mathcal{S}_2$ it simply inserts the item in $\mathcal{S}$ and in the corresponding counter it stores as estimated frequency its frequency in $\mathcal{S}_2$. 

 This could entail the use of up to $2k - 2$ counters in the worst case, when $\mathcal{S}_1$ and $\mathcal{S}_2$ share no item. Let $\mathcal{S}.nz$ be the number of counters in $\mathcal{S}$. The entries in $\mathcal{S}$ are sorted by the counters' frequency in ascending order, and, if $\mathcal{S}.nz \leq k-1$ the algorithm returns $\mathcal{S}$. Otherwise, a pruning operation is required. The combine step can be performed with a constant number of sorts and scans of summaries of size $O(k)$. Then, the algorithm subtracts from the last $k-1$ counters the frequency of the $(S.nz-k+1)$--th counter, removes the first $S.nz-k+1$ counters and returns the remaining $k - 1$ counters, whose frequency has been corrected. The algorithm requires in the worst case time linear in the total number of counters, i.e., $O(k)$ if implemented as described in \cite{Agarwal} using an hash table. 

In the experiments, we tested our algorithm against the one from Agarwal et al. on both synthetic and real datasets. Regarding synthetic datasets, the input distributions used in our experiments are the Riemann--Hurwitz distribution (Hurwitz for short), and its particular case, the Zipf distribution, which is one of the most used in experiments related to sequential algorithms for frequent items. We recall that the Zipf distribution has associated the probability density function (p.d.f.)

\begin{equation}
P_{Z}(x) = \frac{x^{-(\rho +1)}}{\zeta(\rho +1)}  \quad   x \geq 1,
\end{equation}

\noindent where $\rho$ is a positive real parameter controlling the skewness of the distribution and

\begin{equation}
\zeta(s)=\sum_{k=1}^{\infty} \frac{1}{k^{s}} , \qquad \text{Re s}>1
\end{equation}

\noindent is the Riemann zeta function \cite{iwaniec-kowalski}. The Hurwitz distribution has p.d.f.

\begin{equation}
\label{hd}
P_{H}(x, a) = \frac{x^{-(\rho +1)}}{\zeta_H(\rho +1,a)}  \quad   x \geq 1,
\end{equation}

\noindent where

\begin{equation}
\zeta_H(s, q)=\sum_{k=1}^{\infty} \frac{1}{(k+q)^{s}} , \qquad \text{Re s}>1, \quad \text{Re q} >0.
\end{equation}

\noindent is the Riemann--Hurwitz zeta function. Both functions play a crucial role in analytic number theory \cite{iwaniec-kowalski} \cite{Tempesta}. 

The real datasets we used come from different domains \cite{Dallachiesa}. All of the datasets are publicly available, and two of them (Kosarak and Retail) have been widely used and reported in the data mining literature. Overall, the four datasets are characterized by a diversity of statistical characteristics, which we report in Table \ref{data}. 

\textbf{Kosarak:} This is a click-stream dataset of a Hungarian online news portal. It has been anonymized, and consists of transactions, each of which is comprised of several integer items. In the experiments, we have considered every single item in serial order.

\textbf{Retail:} This dataset contains retail market basket data coming from an anonymous Belgian store. Again, we consider all of the items belonging to the dataset in serial order.

\textbf{Q148:} Derived from the KDD Cup 2000 data, compliments of Blue Martini, this dataset contains several data. The ones we use for our experiments are the values of the attribute “Request Processing Time Sum” (attribute number 148), coming from the ``clicks'' dataset. A pre-processing step was required, in order to obtain the final dataset. We had to replace all of the missing values (appearing as question marks) with the value of 0.

\textbf{Nasa:} Compliments of NASA and the Voyager 2 Triaxial Fluxgate Magnetometer principal investigator, Dr. Norman F. Ness, this dataset contains several data. We selected the “Field Magnitude (F1)” and “Field Modulus (F2)” attributes from the Voyager 2 spacecraft Hourly Average Interplanetary Magnetic Field Data. A pre-processing step was required for this dataset: having selected the data for the years 1977-2004, we removed the unknown values (marked as 999), and multiplied all values by 1000 to convert them to integers (since the original values were real numbers with precision of 3 decimal points). The values of the two attributes were finally concatenated. In our experiments, we read all of the values of the attribute “F1”, followed by all of the values of the attribute “F2”.

\begin{table}
\renewcommand{\arraystretch}{1.3}
 \caption{Statistical characteristics of the real datasets}
      \label{data}
	\centering
    \begin{tabular}{|c |  c |  c  | c | c |}
    \hline
      & Kosarak & Retail & Q148 & Nasa  \\ \hline
      \hline
    Count &  8019015 & 908576 & 234954 & 284170 \\ \hline
    Distinct items & 41270 & 16470 & 11824 & 2116 \\ \hline
    Min & 1 & 0 & 0 &  0  \\ \hline 
    Max & 41270 & 16469 & 149464496 & 28474 \\ \hline 
    Mean & 2387.2 & 3264.7 & 3392.9 & 353.9  \\ \hline
    Median & 640 & 1564 & 63 & 120 \\ \hline
    Std. deviation & 4308.5 & 4093.2 & 309782.5 & 778.1 \\ \hline
    Skewness & 3.5 & 1.5 & 478.1 & 6.5  \\ \hline
    \end{tabular}
    \end{table}

Denoting with $f$  the true frequency of an item and with $\hat{f}$ the corresponding frequency reported by an algorithm, then the absolute error is, by definition, the difference $\left|  f - \hat{f} \right|$. The (absolute) total error is then defined as the sum of the absolute errors related to the items reported by an algorithm. Similarly, the absolute relative error is defined as $\Delta f = \frac{{\left| {f - \hat{f}} \right|}}{f}$, and the average relative error is derived by averaging the absolute relative errors over all of the measured frequencies.

Precision is defined as the total number of true $k$-majority elements reported over the total number of items reported. Therefore, this metric quantifies the number of false positives reported by an algorithm in the output data summary. Recall is instead the total number of true $k$-majority elements reported over the number of true $k$-majority elements given by an exact algorithm. It follows that an algorithm is correct if an only if its recall is equal to 1 (or 100\%); both algorithms under test have already been proved to be formally correct and their recall in all of the tests is indeed equal to 1.

\subsection{Real datasets: error}

In this Section, we report the experimental results obtained on the real datasets. We do not report on the performances, owing to the fact that processing the largest dataset on a single processor requires just a few milliseconds. Since the datasets are real, the only parameter we can vary is $k$. In the following tests, $k$ has been varied from 100 to 1000, in steps of 100 (owing to the statistical characteristics of the real datasets). 

We report the total error, the precision and the average relative error (denoted from now on as ARE). 

Figure \ref{tekn} presents the results related to the total error for the Kosarak and Nasa datasets, whilst Figure \ref{teqr} is related to the Q148 and Retail datasets. Note that we use a logarithmic scale for the total error values, since some of the curves would otherwise be too close to distinguish them. As shown, our algorithm outperforms the Agarwal et al. algorithm for all of the datasets under test, with very low and close to zero total error for both the Kosarak and Q148 datasets. The values for both the Nasa and the Retail datasets are about an order of magnitude smaller than the corresponding values obtained by Agarwal et al. 

Regarding the precision, the results in Figures \ref{preckn} (Kosarak and Nasa datasets) and \ref{precqr} (Q148 and Retail datasets) are also clear evidence of the superiority of our algorithm. As shown, the algorithm by Agarwal et al. achieves a precision almost equal to zero for all of the datasets under test. Our algorithm exhibits a precision close to one for both Kosarak and Q148. For the Nasa dataset, the precision is between 0.55 and 0.85 for $k$ in the range [100 - 300], and steadily increases towards 1.0 for $k$ in the range [400 - 1000]. Similarly, for the Retail dataset, the precision is between 0.5 and 1.0 for $k$ in the range [100 - 400], and steadily increases towards 0.7 for $k$ in the range [500 - 1000].

Finally, Figures \ref{arekn} (Kosarak and Nasa datasets) and \ref{areqr} (Q148 and Retail datasets) are related to the average relative error, with our algorithm clearly outperforming the other. Our algorithm exhibits ARE values close to zero for both Kosarak and Q148. For the Nasa dataset, our algorithm's ARE values are steadily decreasing from 0.2 to 0. The same behavior is observed for the Retail dataset, where our algorithm exhibits ARE close to zero for $k = 100$, close to one for $k = 200$, equal to 0.5 for $k = 300$ and then steadily decreasing ARE values from 0.62 to 0.26 in the range [300 - 1000].

\begin{figure*}[hbt]
  \centering
  \begin{tabular}{ c c }
	\subfloat[Kosarak and Nasa datasets]{
           \includegraphics[scale=0.5]{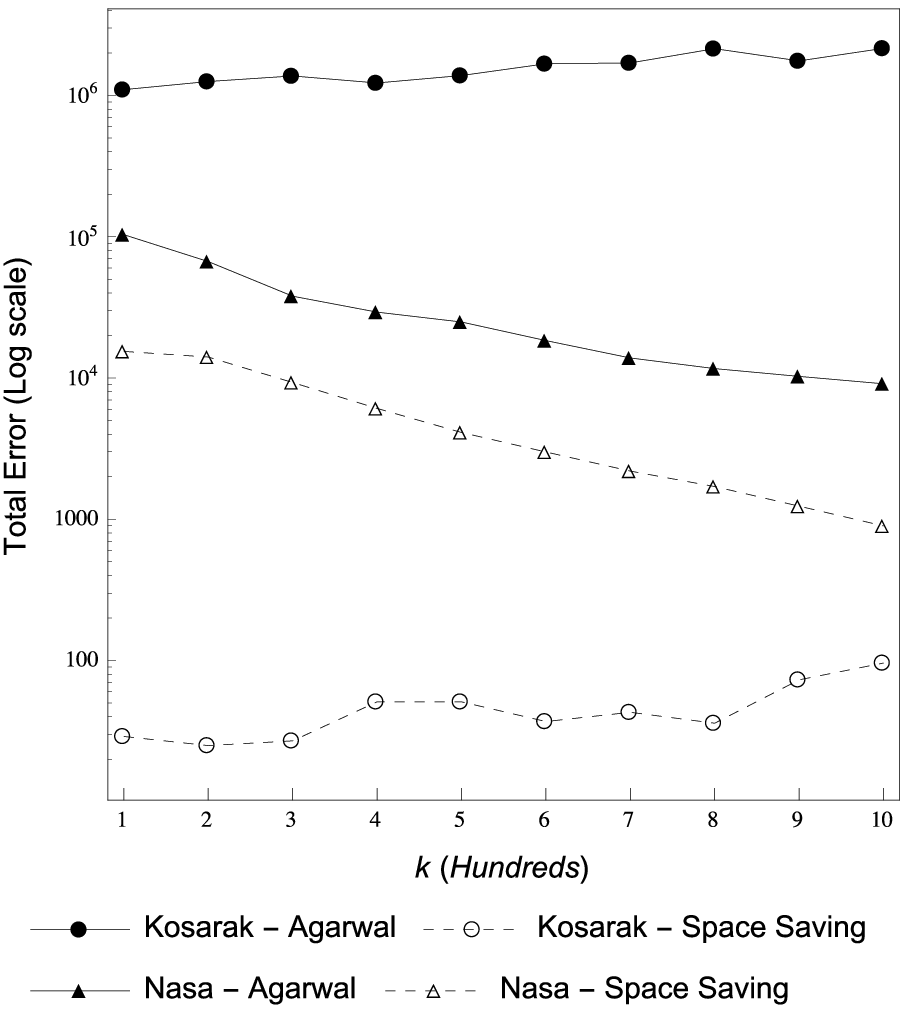}
           \label{tekn}
    } & 
           
    \subfloat[Q148 and Retail datasets]{
          \includegraphics[scale=0.5]{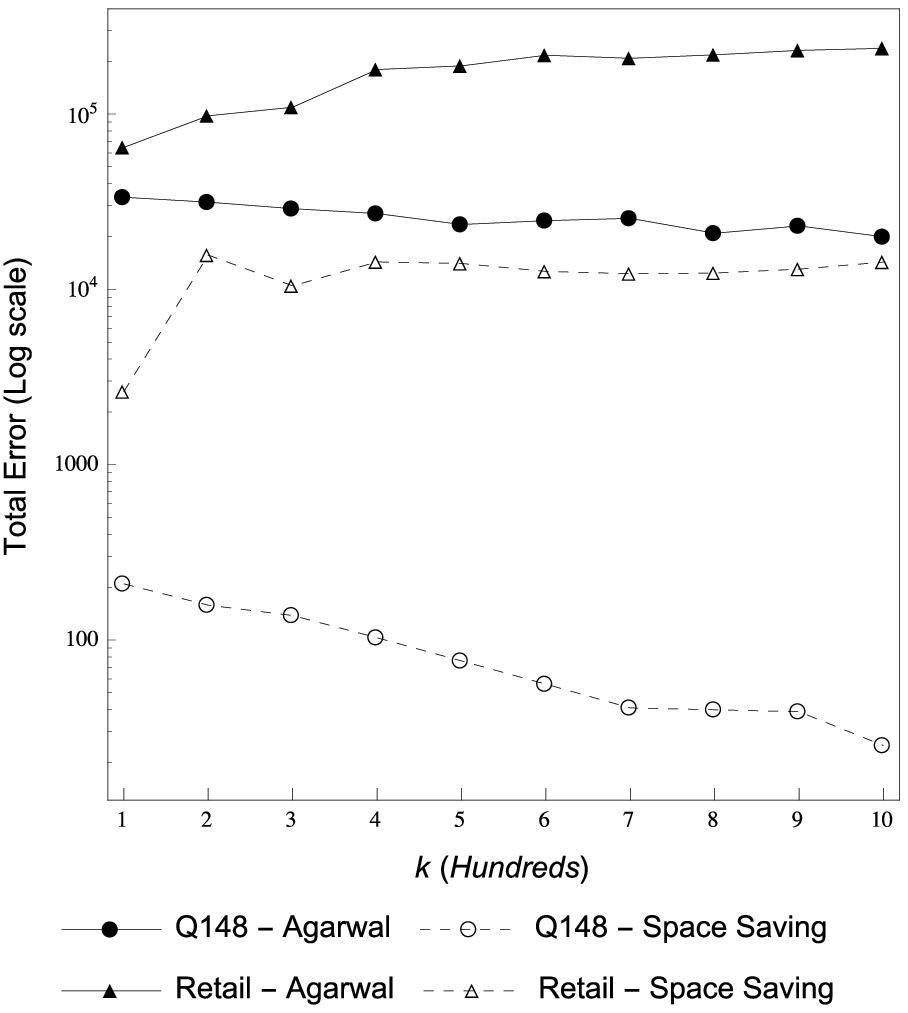}
          \label{teqr}
    } 
  	
\end{tabular}
           
 \caption{Real datasets: Total Error varying $k$ on $p = 8$ cores}
 \label{real-te}
 
\end{figure*}

\begin{figure*}[h!]
  \centering
  \begin{tabular}{ c c }
		
	\subfloat[Kosarak and Nasa datasets]{
           \includegraphics[scale=0.5]{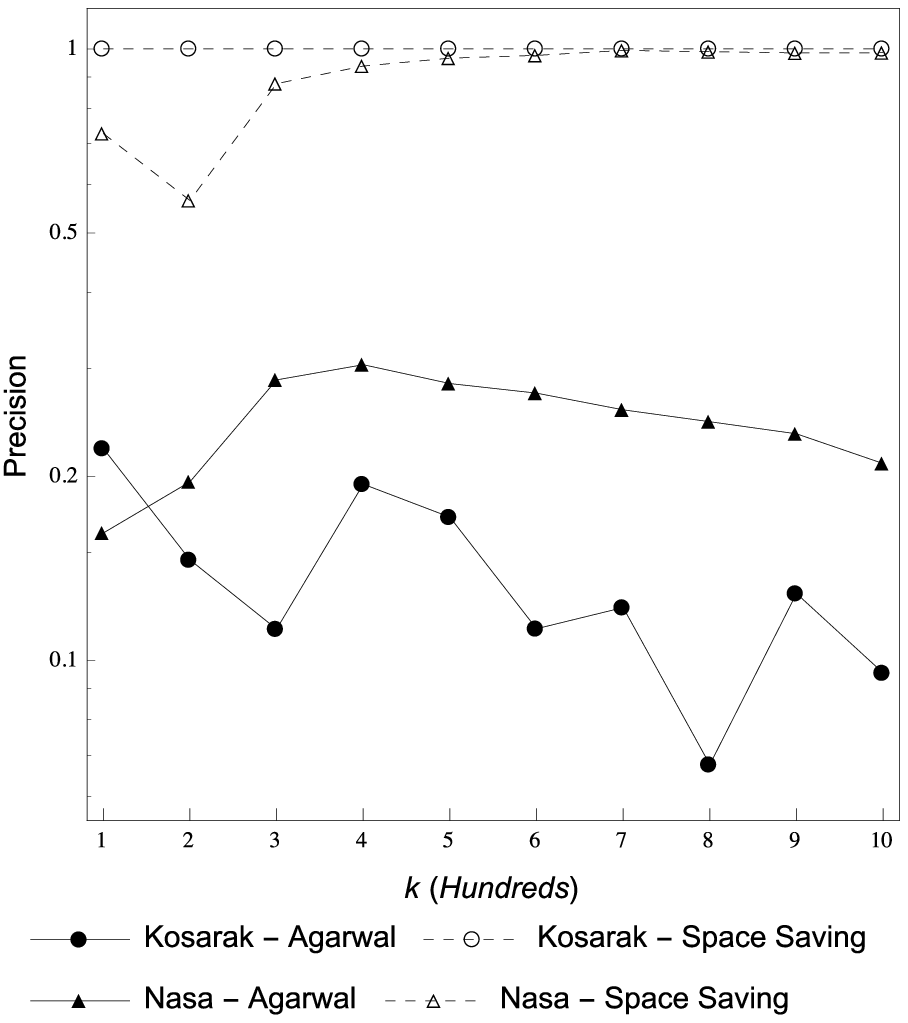}
           \label{preckn}
    } & 
           
    \subfloat[Q148 and Retail datasets]{
          \includegraphics[scale=0.5]{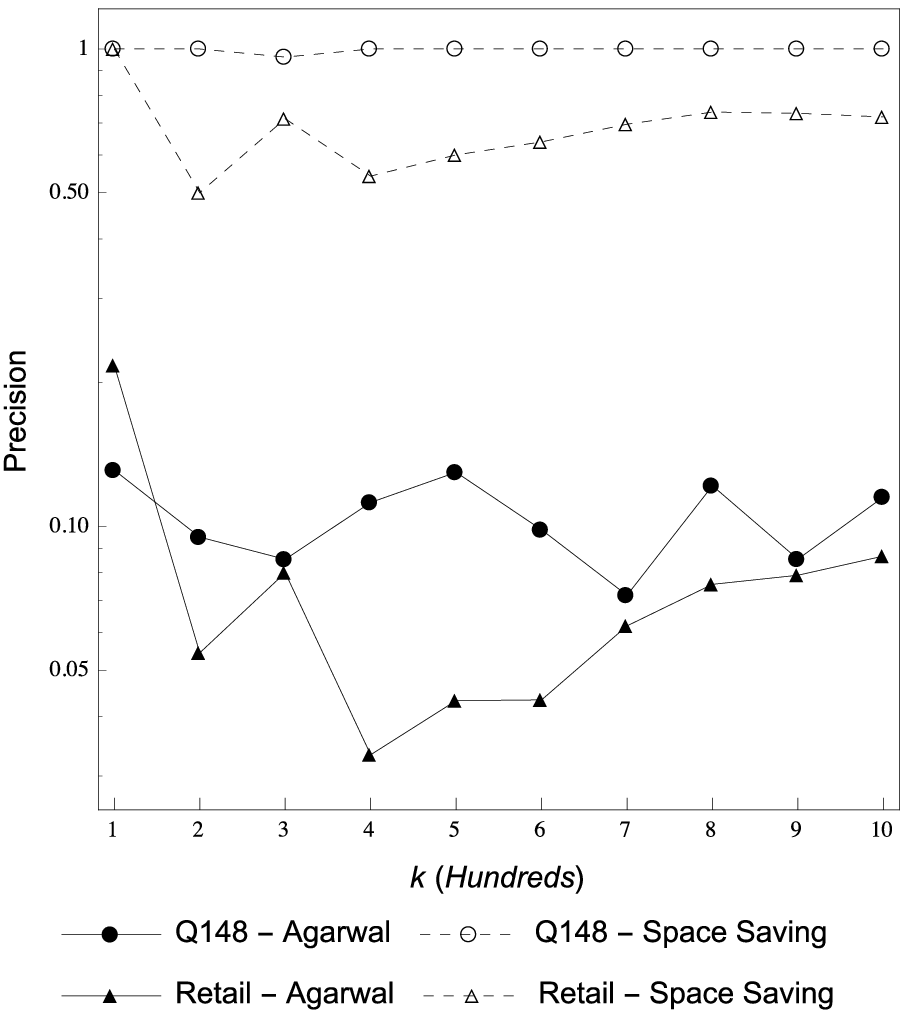}
          \label{precqr}
    }

\end{tabular}
           
 \caption{Real datasets: Precision varying $k$ on $p = 8$ cores}
 \label{real-prec}
 
\end{figure*}

\begin{figure*}[h!]
  \centering
  \begin{tabular}{ c c }
	          
    \subfloat[Kosarak and Nasa datasets]{
           \includegraphics[scale=0.5]{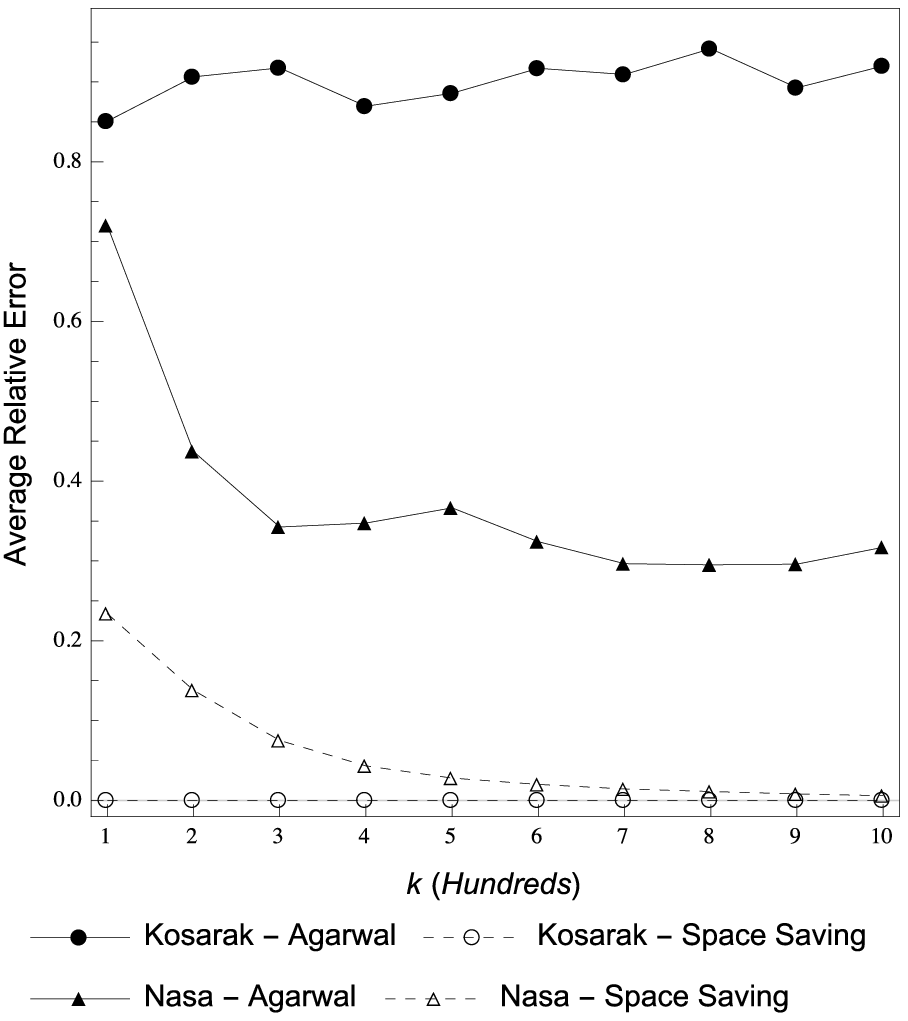}
           \label{arekn}
     } & 
           
    \subfloat[Q148 and Retail datasets]{
          \includegraphics[scale=0.5]{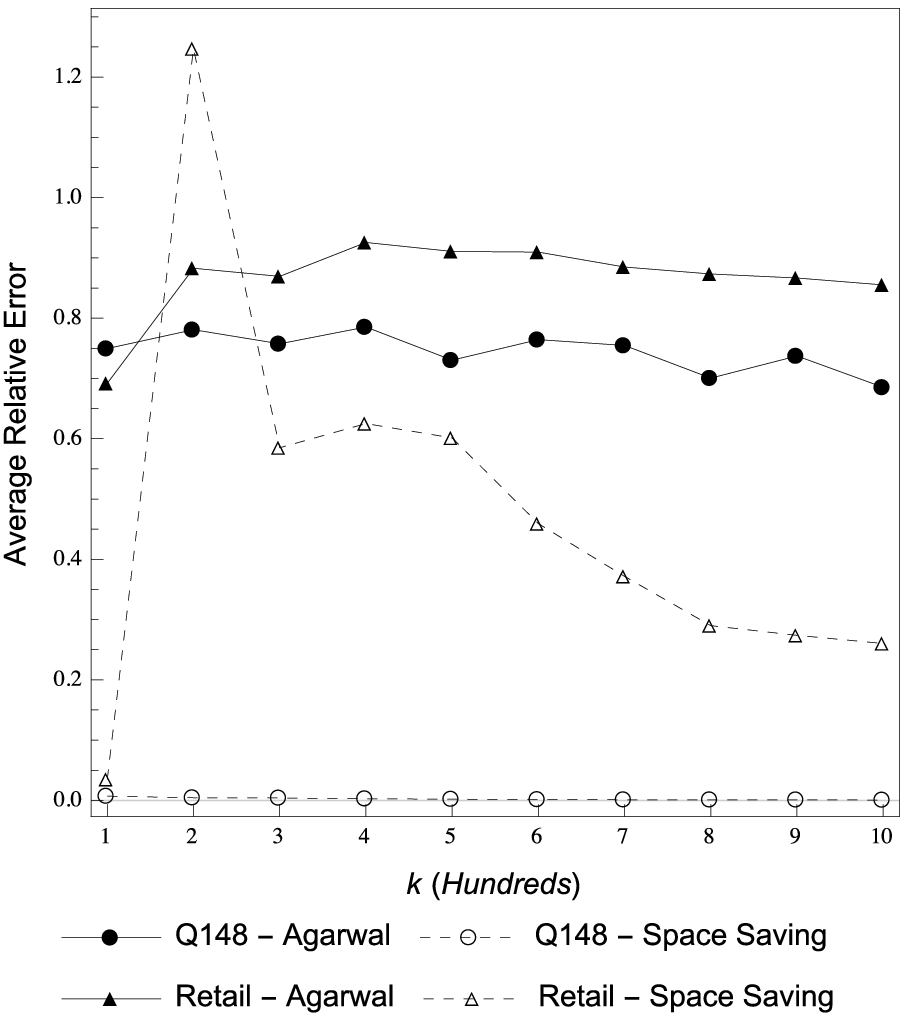}
          \label{areqr}
     }

\end{tabular}
           
 \caption{Real datasets: Average Relative Error varying $k$ on $p = 8$ cores}
 \label{real-are}
 
\end{figure*}

\subsection{Synthetic datasets: error}
\label{error}
We have carried out several experiments, with the aim of analyzing the error committed by the algorithms under test. We have fixed $a=0.5$ in all of the simulations involving the Hurwitz distribution. Indeed, for integer values of the parameter $a$, the Hurwitz distribution becomes the Zipf one (with a shifted value of the $\rho$ parameter). As usual, we report the total error, precision and ARE. 

The following experiments related to the error are characterized by the input size $n$, the parameter $k$ and the skew $\rho$ of the distribution; for each experiment we have determined the corresponding total error,  average relative error and precision. In particular, in the first experiment we fixed $n = 500,000,000$ and $\rho = 1.5$ letting $k$ vary from 1000 to 10,000 in steps of 1000. In the second experiment, $\rho = 1.5$, $k = 2000$ and $n$ varies from 100,000,000 to 1,000,000,000 in steps of 100,000,000. Finally, in the third experiment we fixed $n = 500,000,000$, $k = 2000$ and $\rho$ varies from 0.5 to 3.0 in steps of 0.5. Table \ref{error-experiments} recaps the experiments carried out. For each different value of $k$, $n$ and $\rho$ the algorithms have been run 20 times using a different seed for the pseudo-random generator associated to the distribution (using the same seeds in the corresponding executions of different algorithms). For each input distribution generated, the algorithm has been run on up to 8 cores (one core per node), and the results have been averaged for each number of cores, over all of the runs. The input elements are 32 bits unsigned integers.

We also computed for each mean the corresponding mean's 95\% confidence interval (by using the Student $t$ distribution). Even though we have determined the total error, ARE and precision for each different value of $p = 1, \dots, 8$, we only report here the results for $p = 8$ to save space, taking into account that the observed behavior did not change for $p = 2, \dots, 7$ (and, of course, the behavior for $p = 1$ was identical for both algorithms since no parallel reduction actually took place).

\begin{table}
\renewcommand{\arraystretch}{1.3}
 \caption{Design of  error experiments for Zipfian and Hurwitz distributions }
      \label{error-experiments}
	\centering
	\tiny
	    \begin{tabular}{|c|c|c|c|}
    \hline
    Experiment & $n \quad (millions)$ & $k \quad (thousands)$ & $\rho$  \\ \hline
    1 &  500 &[1, 10] in steps of 1 & 1.5 \\ \hline
    2 &  [100 , 1000] in steps of 100 & 2 & 1.5 \\ \hline
    3 &  500 & 2 & [0.5, 3.0] in steps of 0.5 \\ \hline
    \end{tabular}
    \normalsize
    \end{table}

We begin with the analysis of the total error. For Experiment 1, as shown in Figure \ref{tek}, the total error committed by our algorithm for both input distributions is practically zero for every value of $k$, whilst the total error of the algorithm by Agarwal et al. decreases when $k$ increases but still attains a very high value even for $k = 10,000$. Regarding Experiment 2, depicted in Figure \ref{ten}, again our algorithm is affected by a total error close to zero for both input distributions independently of the value of $n$. On the contrary, the total error of the algorithm by Agarwal et al. steadily increases with $n$ and is already very high even for the smallest value of $n$. In Experiment 3, for both input distributions as shown in Figure \ref{tesk}, our algorithm is affected by total error close to zero. The algorithm by Agarwal et al. on the other hand, performs well only for skew values in the set $\{2.5, 3\}$, whilst the total error explodes for values in the set $\{1, 1.5, 2\}$, attaining its maximum value for $\rho = 1$. To recap, our algorithm outperforms the other with regard to the total error in all of the experiments.

Regarding the ARE, as shown in Figures \ref{arek}, \ref{aren} and \ref{aresk}, our algorithm clearly outperforms the other algorithm for both the input distributions, with ARE values practically equal to zero for the whole set of $k$ and $n$ values under test in Experiments 1 and 2. For Experiment 3, our algorithm shows an ARE value slightly greater than zero  only for $\rho = 0.5$; however, it's worth noting here that $\rho = 0.5$ is the only case in which there are no frequent items.

Finally, we analyze the precision attained. As shown in Figures \ref{preck} and \ref{precn}, in the Experiments 1 and 2 our algorithm clearly outperforms the algorithm by Agarwal et al. for both the input distributions. We obtain precision values equal to one for the whole set of $k$ and $n$ values under test, whilst the Agarwal et al. algorithm's precision is always less than 0.1. For Experiment 3, depicted in Figure  \ref{precsk}, our algorithm provides excellent performances with precision equal to one for skew values in the set  $\{1, 1.5, 2, 2.5, 3\}$.  We note here that the precision is zero for both algorithms when $\rho = 0.5$, which is consistent with our previous observation (when discussing the ARE values) that in this case there are no frequent items. The precision obtained using the algorithm by Agarwal et al. reaches its maximum value (less than 0.2 nevertheless) for $\rho = 1$, and then steadily decreases again. Therefore, in each of the different scenarios, the precision provided by our parallel algorithm is for all of the practical purposes identical to the precision attained by the sequential Space Saving algorithm, so that our main goal when designing the algorithm has been achieved.

\begin{figure*}[h!p]
  \centering
  \begin{tabular}{ccc}
     \subfloat[Total Error varying $k$]{
           \includegraphics[scale=0.5]{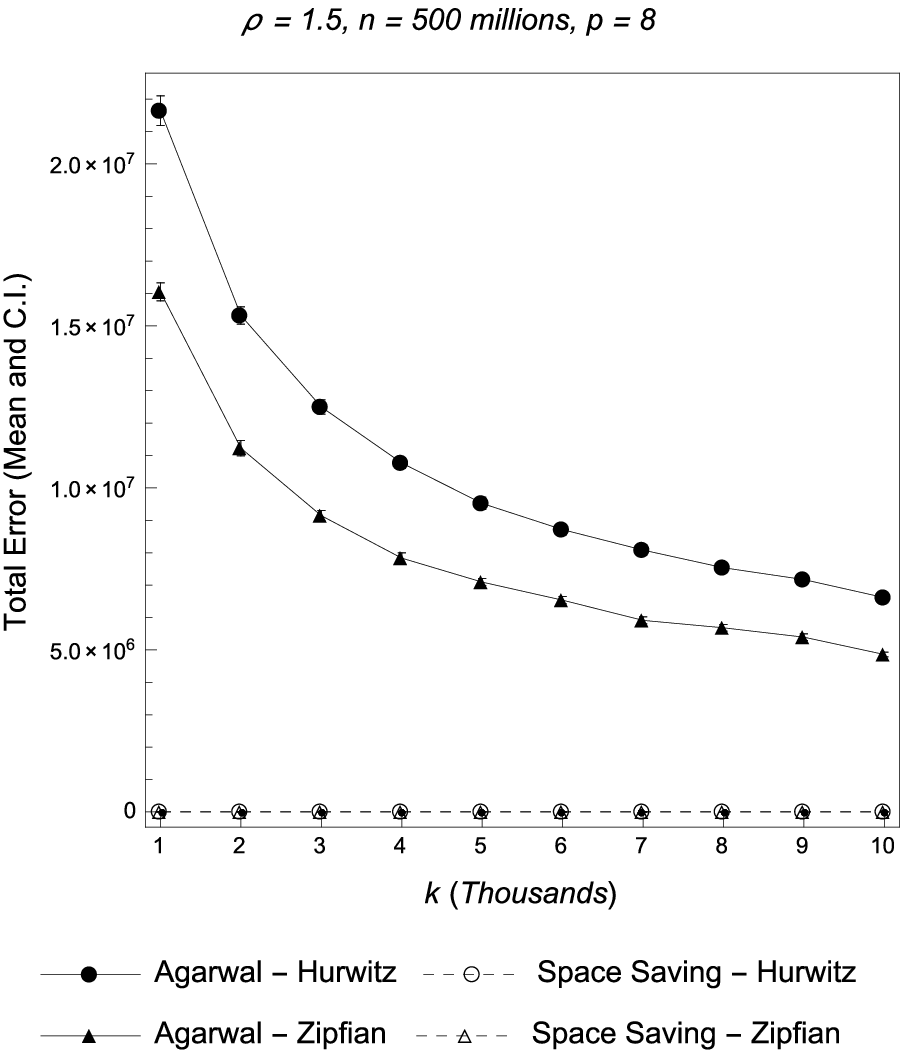}
           \label{tek}
        } &
        
      \subfloat[Precision varying $k$]{
           \includegraphics[scale=0.5]{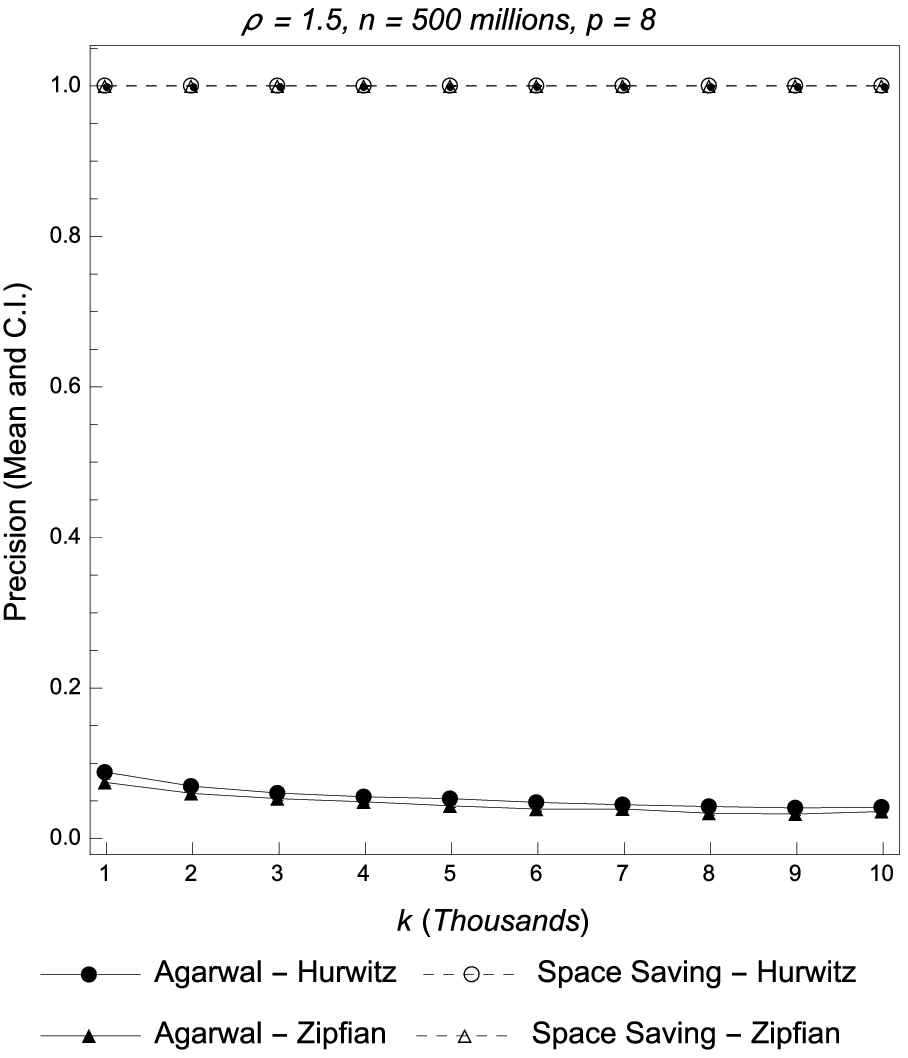}
           \label{preck}
        } &

       \subfloat[ARE varying $k$]{
           \includegraphics[scale=0.5]{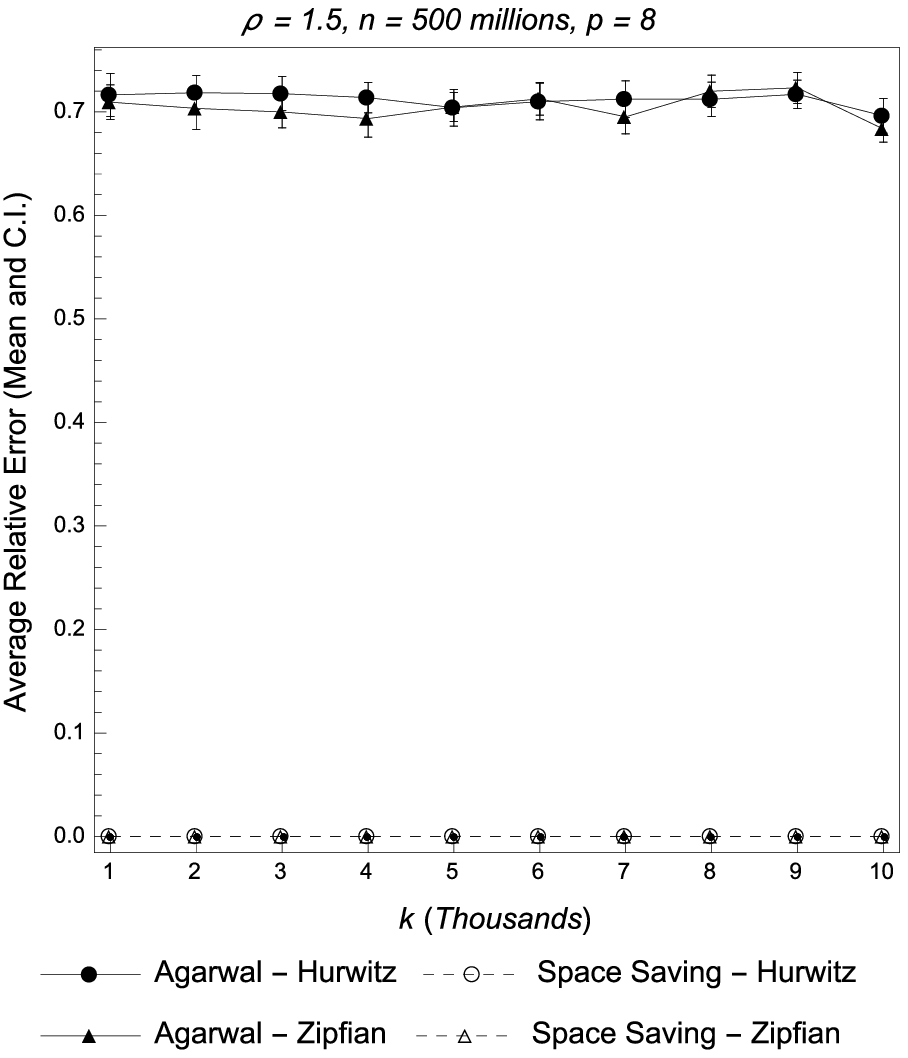}
           \label{arek}
        }

 \end{tabular}
 
 \caption{Experiment 1, Total Error, Precision and ARE varying $k$ on $p = 8$ cores} \label{exp1}
\end{figure*}

\begin{figure*}[h!p]
  \centering
  \begin{tabular}{ccc}
    
     \subfloat[Total Error varying $n$]{
          \includegraphics[scale=0.5]{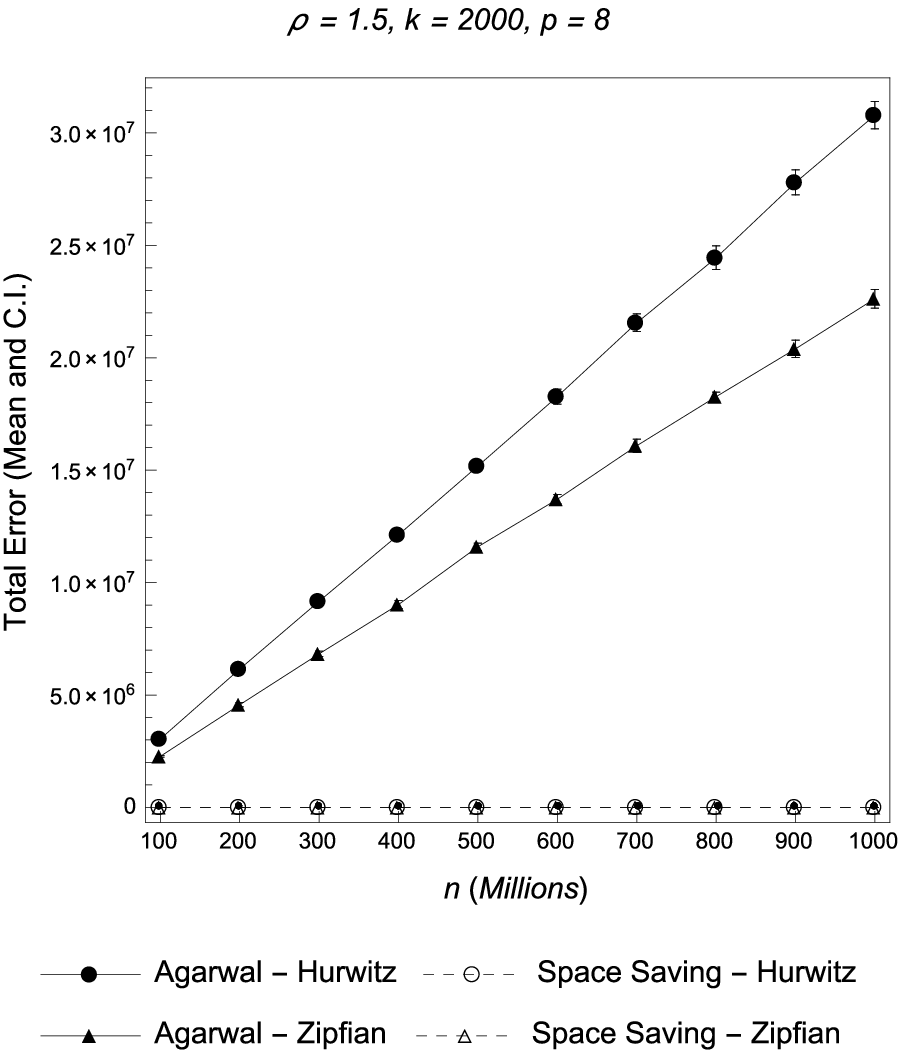}
          \label{ten}
        } &

     \subfloat[Precision varying $n$]{
          \includegraphics[scale=0.5]{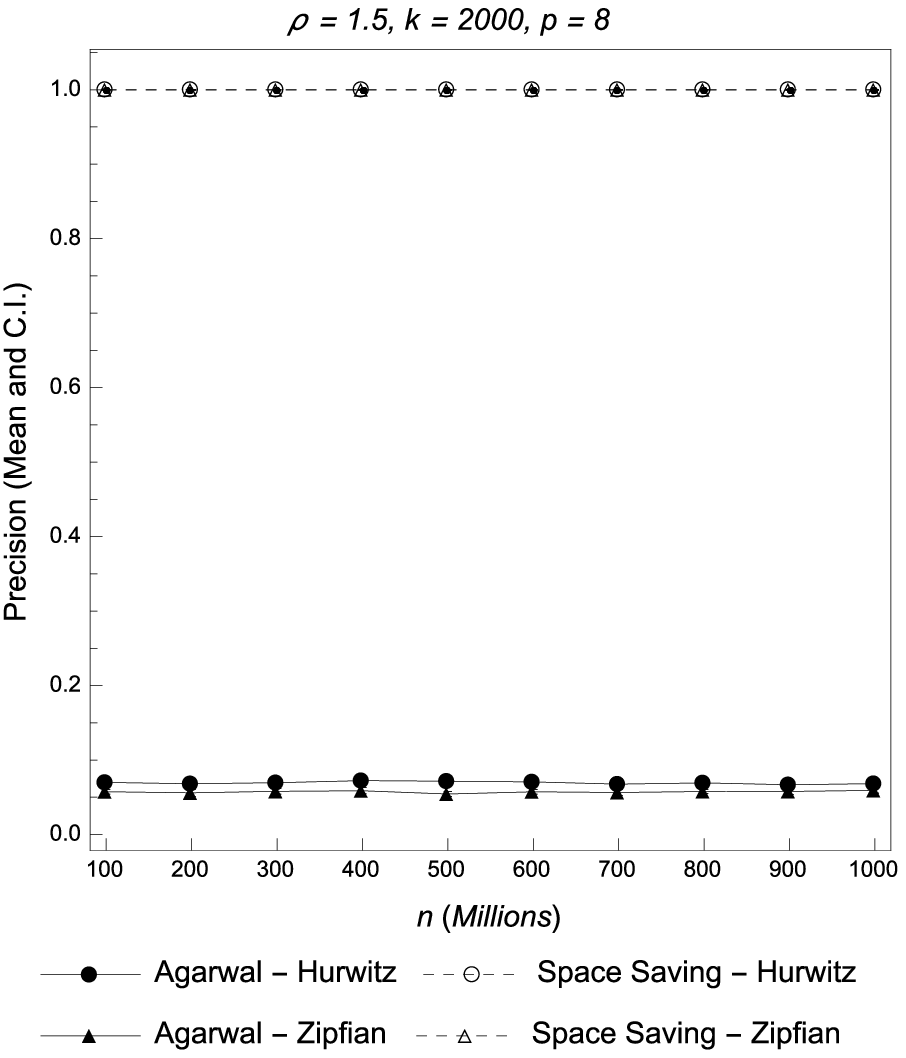}
          \label{precn}
        } &

     \subfloat[ARE varying $n$]{
          \includegraphics[scale=0.5]{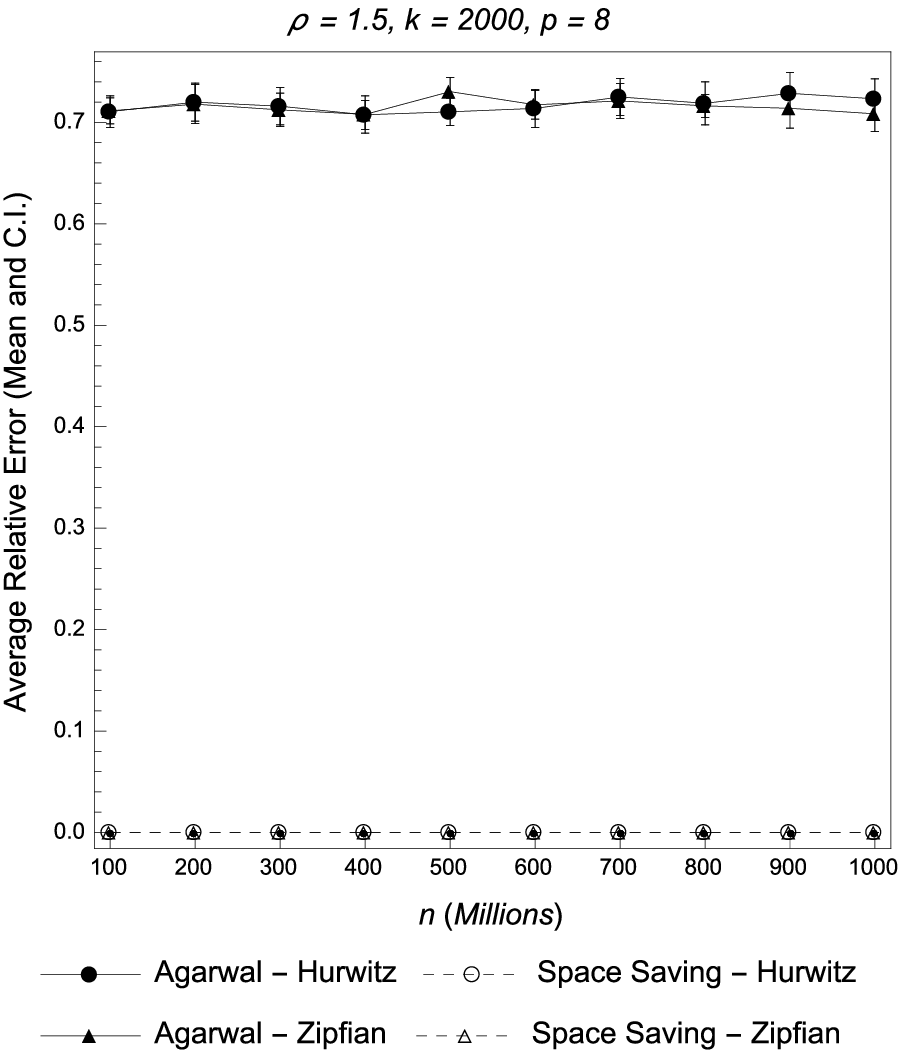}
          \label{aren}
        }
        
 \end{tabular}
 
 \caption{Experiment 2, Total Error, Precision and ARE varying $n$ on $p = 8$ cores} \label{exp2}
\end{figure*}

\begin{figure*}[h!p]
  \centering
  \begin{tabular}{ccc}
            
      \subfloat[Total Error varying $\rho$]{
          \includegraphics[scale=0.5]{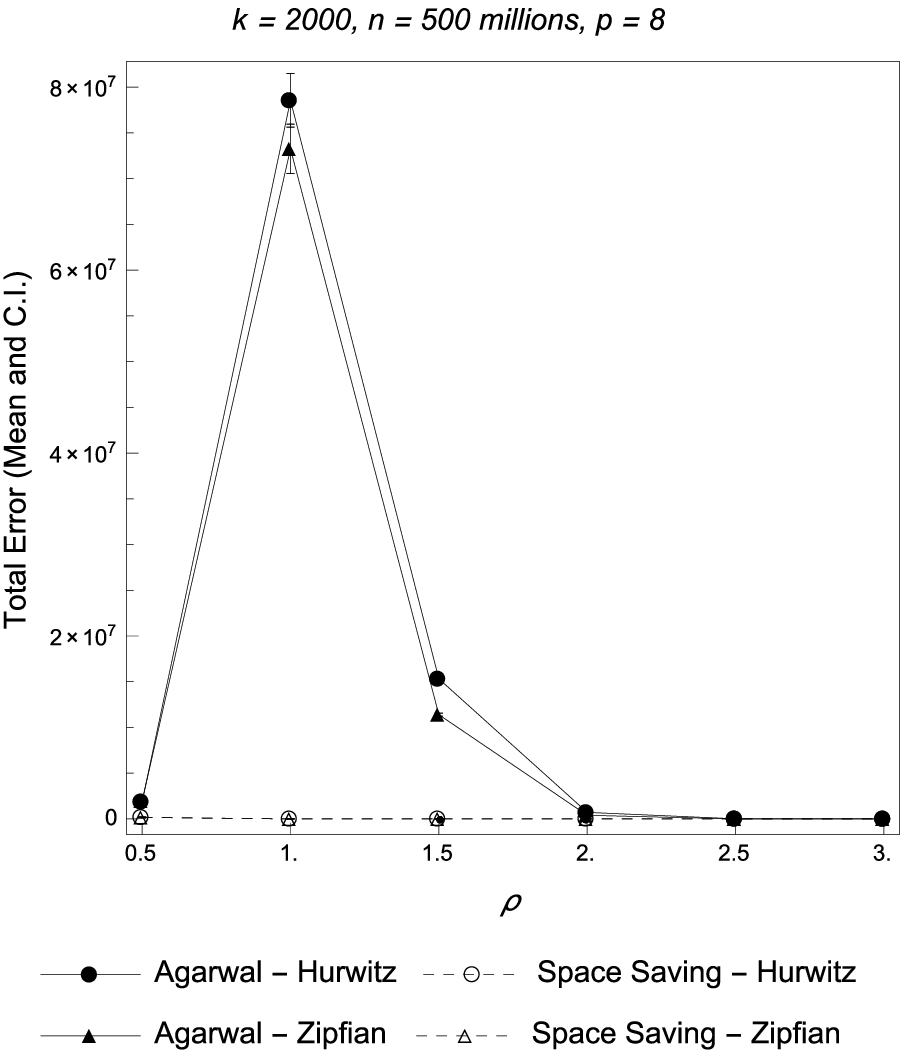}
          \label{tesk}
        } &
        
      \subfloat[Precision varying $\rho$]{
          \includegraphics[scale=0.5]{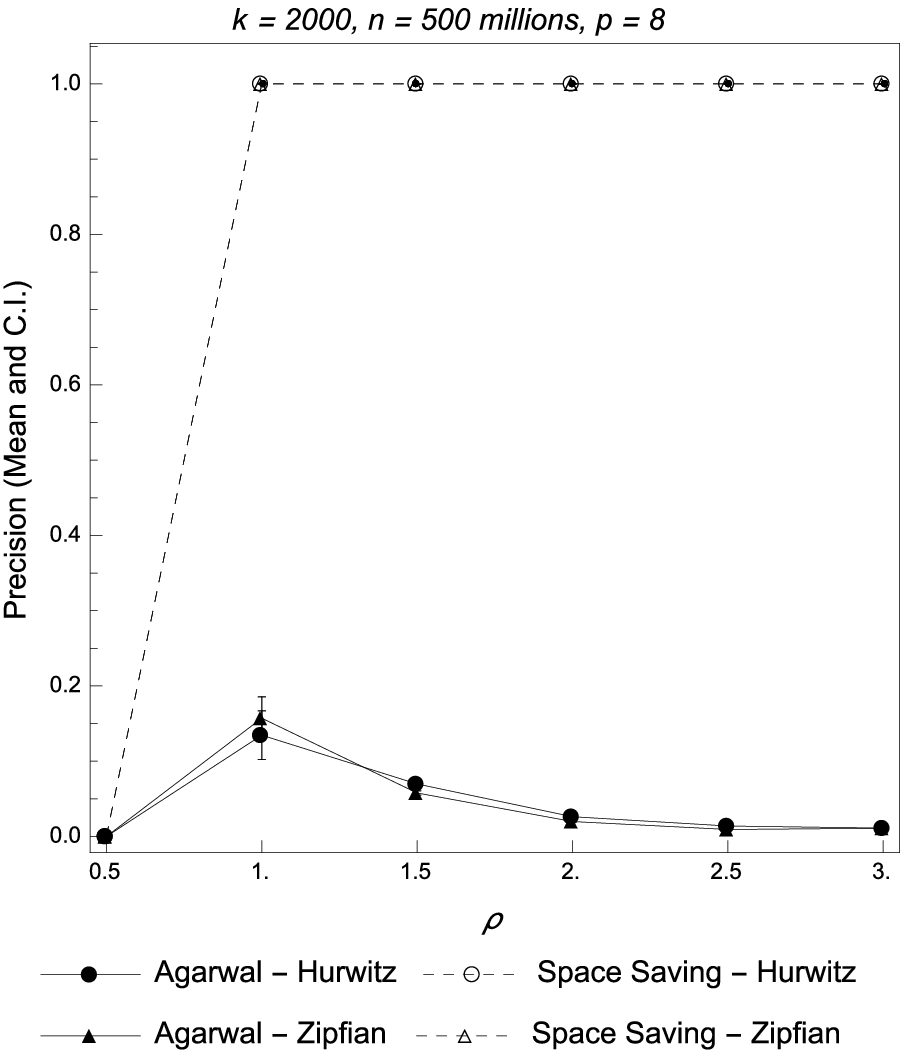}
          \label{precsk}
        } &

     \subfloat[ARE varying $\rho$]{
          \includegraphics[scale=0.5]{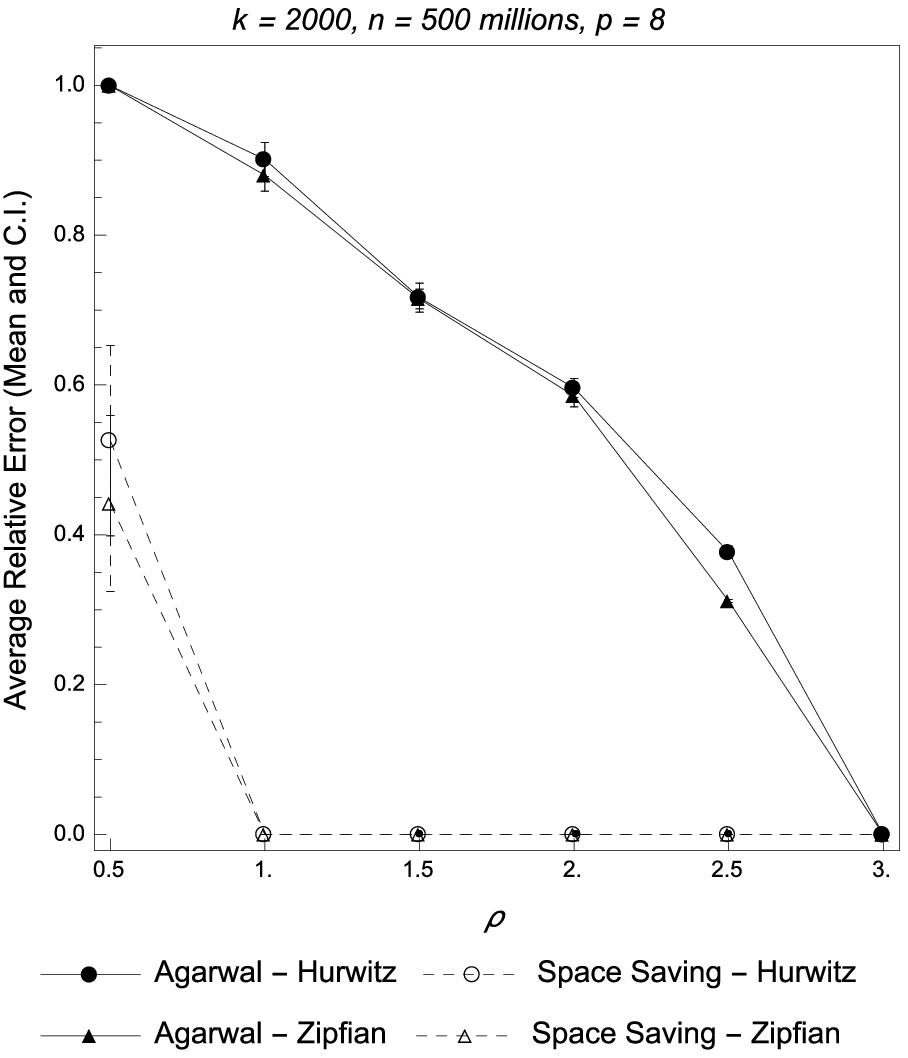}
          \label{aresk}
        }
        
 \end{tabular}
 
 \caption{Experiment 3, Total Error, Precision and ARE varying $\rho$ on $p = 8$ cores} \label{exp3}
\end{figure*}

\subsection{Synthetic datasets: performances}

We have designed and carried out some performance experiments characterized by the following parameters: the input size $n$, $k$ and the skew $\rho$. For each input distribution generated, the algorithm has been run twenty times on up to 8 cores, and the results have been averaged for each number of cores, over all of the runs.  The input elements are 32 bits unsigned integers. Table \ref{experiments} reports the values actually used in each of the performance experiments.

\begin{table}
\renewcommand{\arraystretch}{1.3}
 \caption{Design of  performance experiments for Zipfian and Hurwitz distributions }
      \label{experiments}
	\centering
    \begin{tabular}{|c |  c |  c  | c | }
    \hline
    Experiment & $n$ & $k$ & $\rho$  \\ \hline \hline
    4 &  4,000,000,000 & 2,000 & 1.5 \\ \hline
    5 &  4,000,000,000 & 3,000 & 3.0 \\ \hline
    \end{tabular}
    \end{table}

 As shown in Table \ref{experiments}, we have fixed $n$ to 4 billions of input items and, in each experiment, we vary the values of $k$ and $\rho$. Figures \ref{exp4} and \ref{exp5}, related to Experiments 4 and 5, respectively, show the performances for both the Zipfian and Hurwitz distributions with regard to running time, speedup and efficiency. Similar results were obtained with other settings.
 
 \begin{figure*}[h!p]
  \centering
  \begin{tabular}{ c c }
     \subfloat[Zipfian]{
           \includegraphics[scale=0.6]{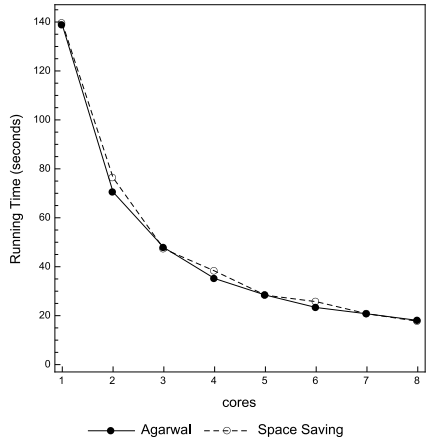}
           \label{z11r}
        } &
        
     \subfloat[Hurwitz]{
          \includegraphics[scale=0.6]{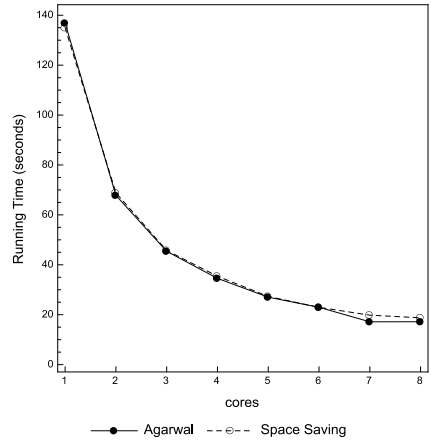}
          \label{h11r}
        } \\

     \subfloat[Zipfian]{

           \includegraphics[scale=0.6]{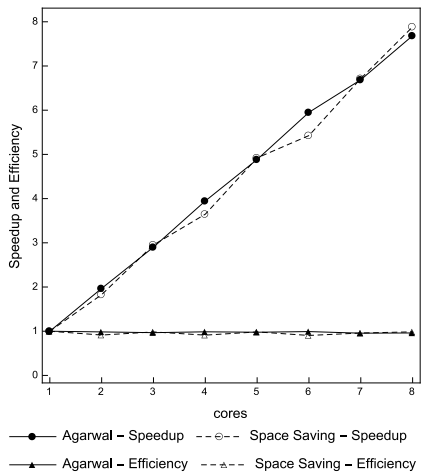}
           \label{z11se}
        } &

     \subfloat[Hurwitz]{
          \includegraphics[scale=0.6]{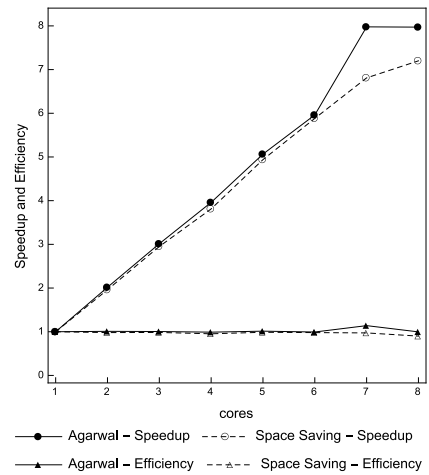}
          \label{h11se}
        }

\end{tabular}

 \caption{Experiment 4: Running Time, Speedup and Efficiency} \label{exp4}
\end{figure*}

\begin{figure*}[h!p]
  \centering
  \begin{tabular}{ c c }
     \subfloat[Zipfian]{
           \includegraphics[scale=0.6]{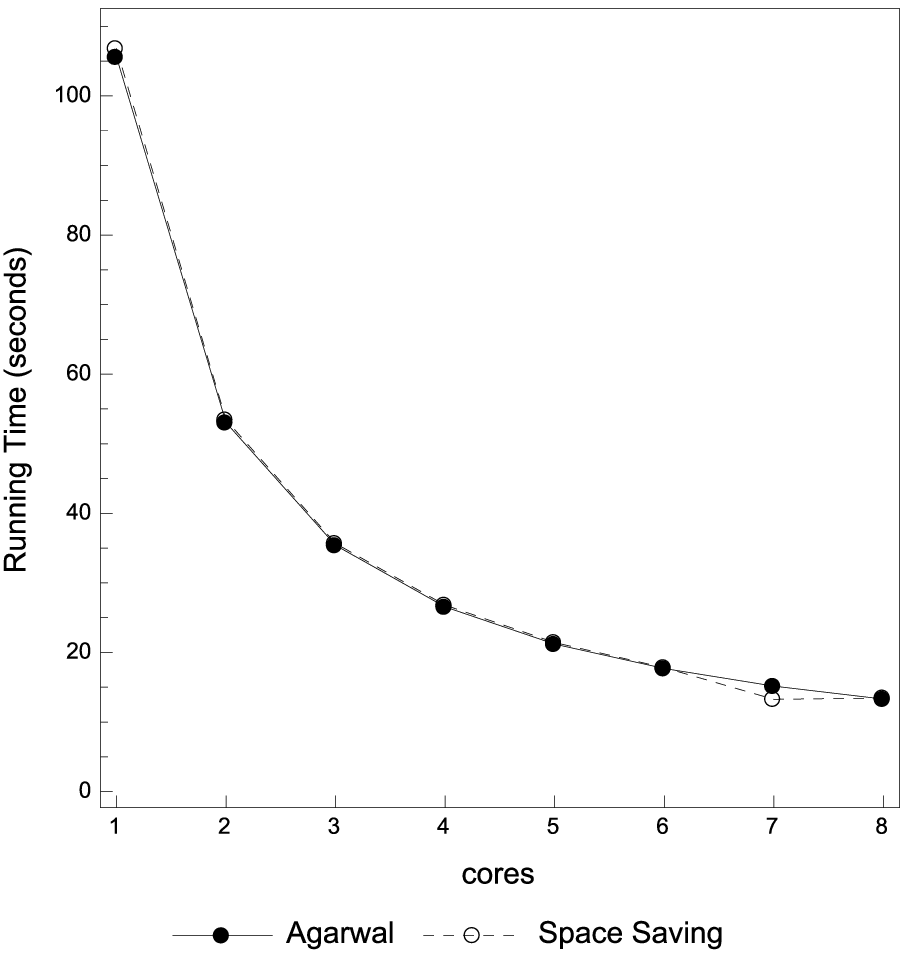}
           \label{z22r}
        } &
      
     \subfloat[Hurwitz]{
          \includegraphics[scale=0.6]{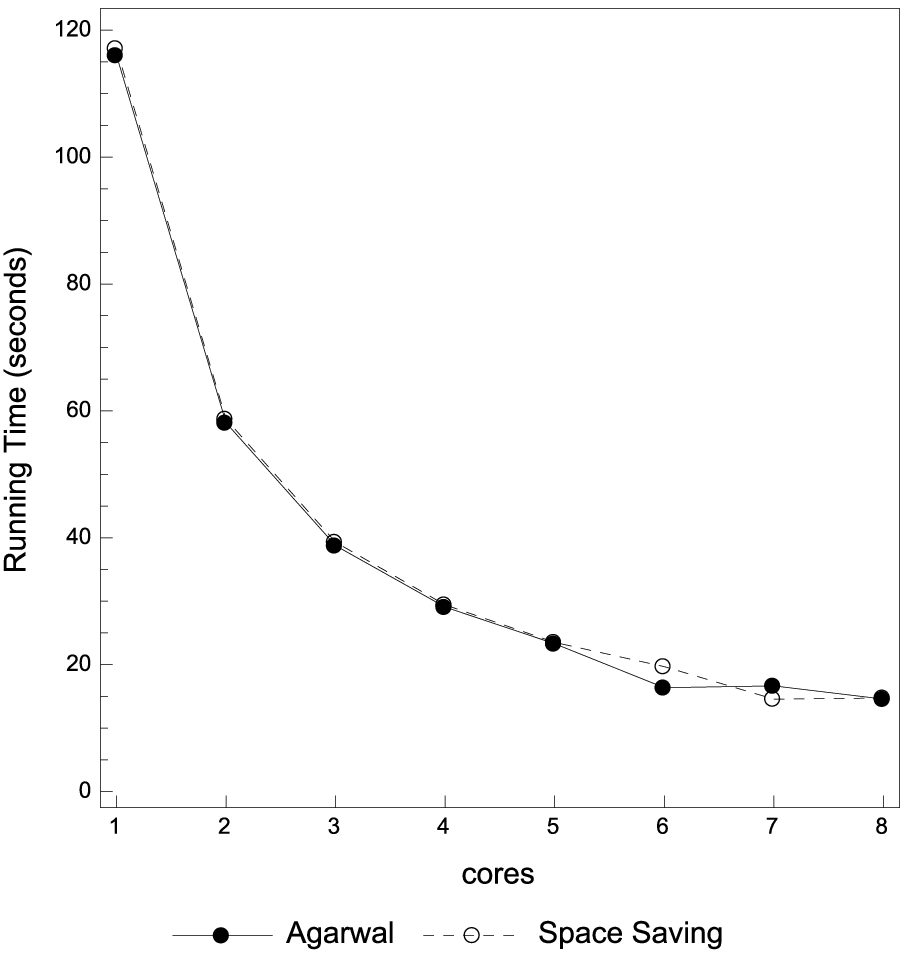}
          \label{h22r}
        } \\

     \subfloat[Zipfian]{
           \includegraphics[scale=0.6]{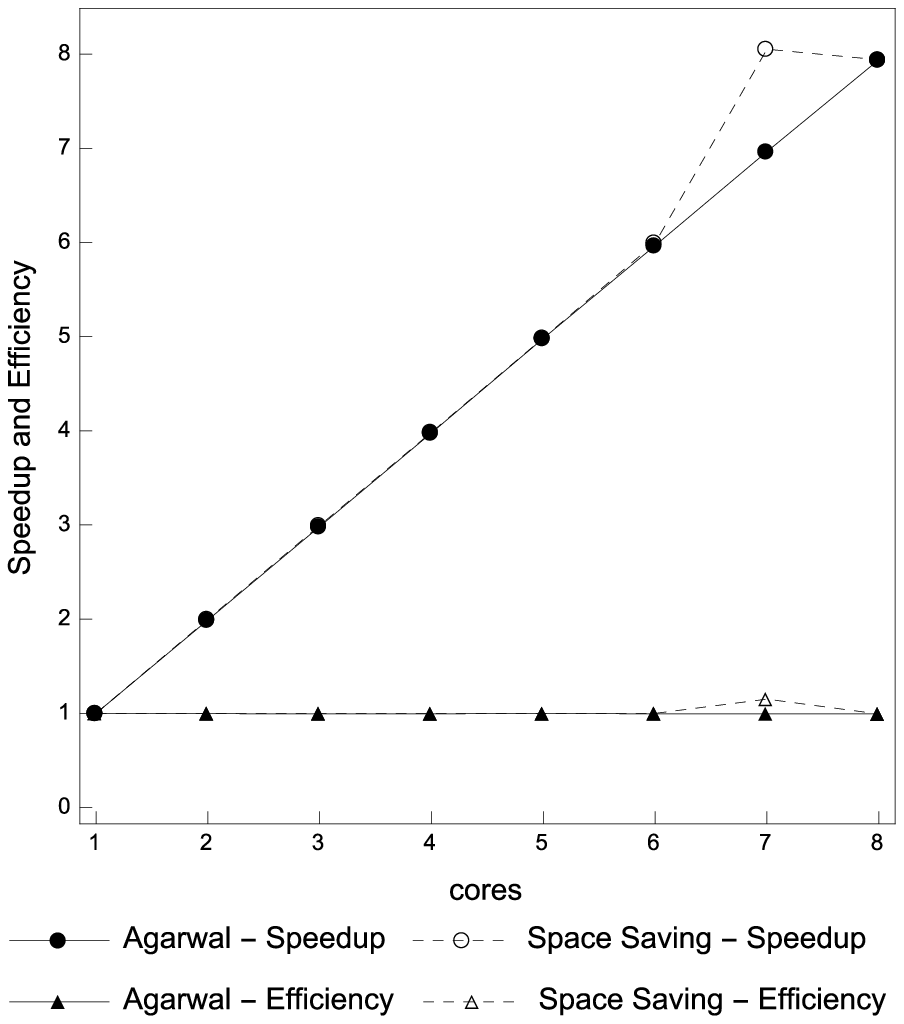}
           \label{z22se}
        } &
    
     \subfloat[Hurwitz]{
          \includegraphics[scale=0.6]{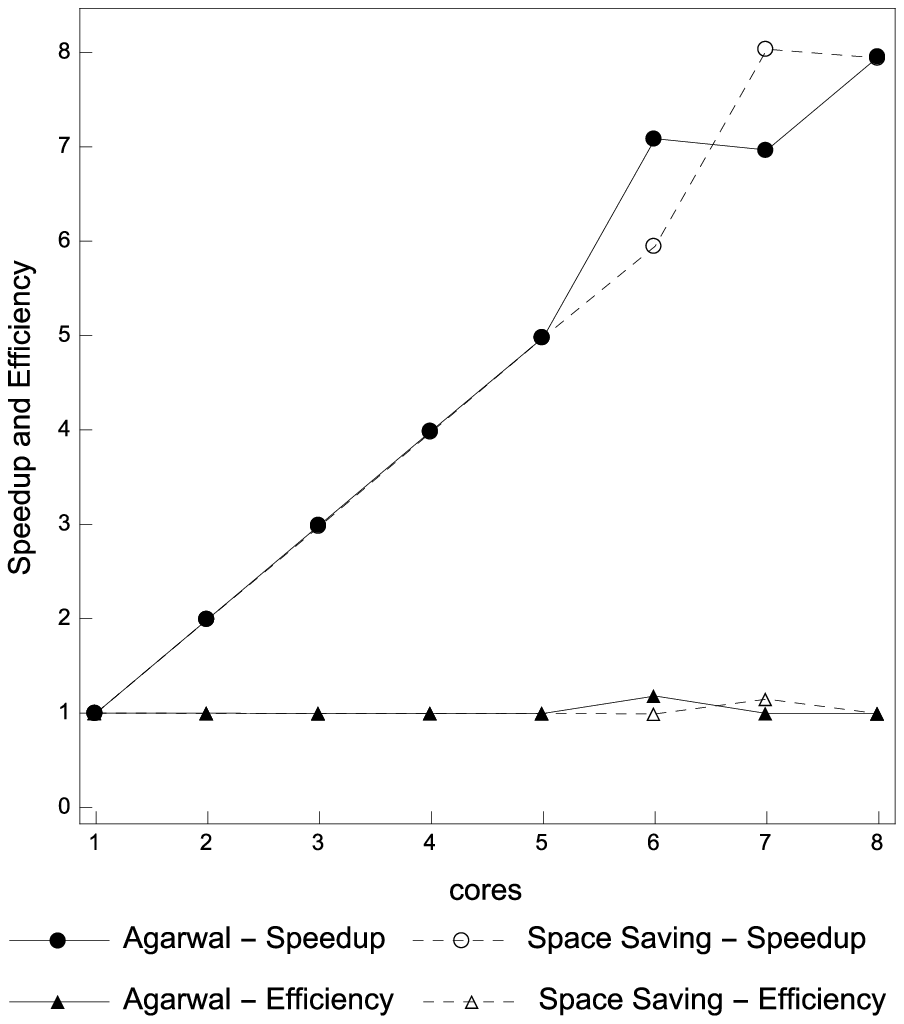}
          \label{h22se}
        }

\end{tabular}

\caption{Experiment 5: Running Time, Speedup and Efficiency} \label{exp5}
\end{figure*}

It is immediate to verify that the performances of our parallel Space Saving algorithm are comparable to the performances of the algorithm by Agarwal et al. for both Zipfian and Hurwitz distributions, with regard to overall running time, speedup and efficiency. In particular, the measured speedup shows in general a linear behavior, with corresponding efficiency close to 1 (or 100\%). It is worth noting here that both algorithms exhibit, in some cases, a slightly superlinear speedup. This phenomenon, observed experimentally, is due to the cluster's memory hierarchy and to related cache effects. So-called superlinear speedups, i.e., speedups which are greater than the number of processors/cores \cite{Bader}, are a source of confusion because in theory this phenomenon is not possible according to  Brent's principle \cite{brent74} (which states that a single processor can simulate a $p$-processor algorithm with a uniform slowdown factor of $p$).

Experimentally, a superlinear  speedup can be observed without violating Brent's principle when the storage space required to run the code on a particular instance exceeds the
memory available on the single-processor machine, but not that of the parallel machine used for the simulation. In such a case, the sequential code needs to swap to secondary
memory (disk) while the parallel code does not, therefore yielding a dramatic slowdown of the sequential code. On a more modest scale, the same problem could occur one level
higher in the memory hierarchy, with the sequential code constantly cache-faulting while the parallel code can keep all of the required data in its cache subsystems. A
sequential algorithm using $M$ bytes of memory will use only $M/p$ bytes on each processor of a $p$ processor parallel system, so that it is easier to keep all of the data in
cache memory on the parallel machine. This is exactly what happened in our simulations.

We recall here that other possible sources of superlinear speedup include some brute--force search problems and the use of a suboptimal sequential algorithm. A parallel system
might exhibit such behavior in search algorithms. In search problems performed by exhaustively looking for the solution, suppose the solution space is divided among the
processors for each one to perform an independent search. In a sequential  implementation the different search spaces are attacked one after the other, while in parallel they
can be done simultaneously, and one processor may find the solution almost immediately, yelding a superlinear speedup.
A parallel system might also exhibit such behavior when using a suboptimal sequential algorithm: each processing element spends less than the time required by the sequential
algorithm divided by  $p$ solving the problem. Generally, if a purely deterministic parallel algorithm were to achieve  better than $p$ times the speedup over the current
sequential algorithm, the parallel algorithm (by Brent's principle) could be emulated on a single processor one parallel part after another, to achieve a faster serial program,
which contradicts the assumption of an optimal serial program.

We recall here that in experiment 5 we used a skew value equal to 3.0, which corresponds to highly skewed distributions of no real practical interest. However, we did these tests anyway for completeness, to test the performances of the algorithms also in this case. Even though both algorithms under test show comparable performances for all of the practical purposes, our algorithm outperforms the one by Agarwal et al. with regard to the error committed, as shown in Section \ref{error}.

\bibliographystyle{elsarticle-num}
\bibliography{bibliography}

\end{document}